\documentclass[11pt]{article}
\pdfoutput=1
\usepackage{odonnell,pdfpages}

\renewcommand{\phi}{\varphi}

\newcommand{\U}{{\calP}}

\newcommand{\Refl}{\mathrm{REFL}_{p}}
\newcommand{\RotF}{\mathrm{ROT}_{y}}
\newcommand{\RR}{\mathcal{U}}

\newcommand{\bonef}{\wt{\bone}_y}

\newcommand{\ketzero}{\ket{\vec{\mathtt{0}}}}
\newcommand{\brazero}{\bra{\vec{\mathtt{0}}}}
\newcommand{\x}{X}

\newcommand{\garbage}{\textnormal{garbage}}

\newcommand{\second}{s}

\DeclareMathOperator{\hav}{hav}

\renewcommand{\i}{\mathrm{i}}
\newcommand{\bii}{\boldsymbol{\mathrm{i}}}

\usepackage{tikz}
\usetikzlibrary{arrows}

\usepackage[framemethod=TikZ]{mdframed}
\newcounter{prob}
\renewcommand{\theprob}{\arabic{prob}}

\newenvironment{prob}[1][]{\refstepcounter{prob}

    \ifstrempty{#1}{\mdfsetup{frametitle={\tikz[baseline=(current bounding box.east),outer sep=0pt,inner sep=2pt]
        \node[anchor=east,rectangle,fill=black!20]
        {\strut Problem~\theprob};}
    }}{\mdfsetup{frametitle={\tikz[baseline=(current bounding box.east),outer sep=0pt,inner sep=2pt]
        \node[anchor=east,rectangle,fill=black!20]
        {\strut Problem~\theprob~(#1)};}}}\mdfsetup{innertopmargin=0pt,linecolor=black!20,linewidth=2pt,topline=true,frametitleaboveskip=\dimexpr-\ht\strutbox\relax }

\begin{mdframed}[]\relax}{\end{mdframed}}
\Crefname{prob}{Problem}{Problems}

\renewcommand{\backref}[1]{}

\renewcommand{\backrefalt}[4]{\ifcase #1 \or
[p.\ #2]\else
[pp.\ #2]\fi}

\begin{document}

\title{Mean estimation when you have the source code; \\
or, quantum Monte Carlo methods}

\author{
Robin Kothari\thanks{Microsoft Quantum. \texttt{robin.kothari@microsoft.com}}
\and
Ryan O'Donnell\thanks{Carnegie Mellon University Computer Science Department.
             Part of this research was performed while the author was at Microsoft Quantum.
             This work was partially supported by ARO grant W911NF2110001.
             \texttt{odonnell@cs.cmu.edu}}
}

\date{\today}

\maketitle

\begin{abstract}
    Suppose $\by$ is a real random variable, and one is given access to ``the code'' that generates it (for example, a randomized or quantum circuit whose output is~$\by$).
    We give a quantum procedure that runs the code~$O(n)$ times and returns an estimate $\widehat{\bmu}$ for $\mu = \E[\by]$ that  with high probability satisfies $\abs{\widehat{\bmu} - \mu} \leq \sigma/n$, where $\sigma = \stddev[\by]$.
    This dependence on~$n$ is optimal for quantum algorithms.
    One may compare with classical algorithms, which can only achieve the quadratically worse $\abs{\wh{\bmu} - \mu} \leq \sigma/\sqrt{n}$.
    Our method improves upon previous works, which either made additional assumptions about~$\by$, and/or assumed the algorithm knew an  a priori bound on~$\sigma$, and/or used additional logarithmic factors beyond~$O(n)$.
    The central subroutine for our result is essentially Grover's algorithm but with complex phases.
\end{abstract}

\section{Introduction}      \label{sec:intro}

Let $\by$ be a real random variable.\footnote{Throughout we use \textbf{boldface} to denote random quantities.}
One may wish to estimate its mean $\mu = \E[\by]$ from independent samples $\by_1, \by_2, \dots, \by_n$.
A natural strategy is to output the sample mean $\widehat{\bmu} = (\by_1 + \cdots + \by_n)/n$, an unbiased estimator with standard deviation~$\sigma/\sqrt{n}$, where $\sigma = \stddev[\by] = \sqrt{\E[(\by-\mu)^2]}$.
Then Chebyshev's inequality implies, say,
\begin{equation}    \label[ineq]{ineq:cheby-example}
    \Pr[\abs{\widehat{\bmu} - \mu} \geq 10\sigma/\sqrt{n}] \leq 1\%.
\end{equation}
As familiar special cases: if $\by$ is bounded in $[0,1]$, then $\sigma \leq 1$ and we get that $n = O(1/\eps^2)$ samples suffice to ensure $|\widehat{\bmu} - \mu| \leq \eps$ with high probability; if $\by \in \{0,1\}$, then $\sigma = \sqrt{\mu(1-\mu)} \leq \sqrt{\mu}$, and we get that $n = O(1/\eps)$ samples suffice to distinguish $\mu \geq \eps$ from $\mu \leq \eps/2$ with high probability.
Up to constant factors, these guarantees cannot be improved upon if the samples $\by_1, \dots, \by_n$ are coming ``from nature''.

\emph{But what if we have ``the code'' for~$\by$?}
By this we mean we have access to, say, a randomized circuit~$C$ whose output is~$\by$.
In a certain sense this means we don't need any samples at all to estimate~$\mu$: By enumerating all possible random paths for~$C$, we can compute~$\mu$ perfectly.
But this could be enormously expensive; if running~$C$ to produce a single sample takes effort~$S$, then the brute-force enumeration analysis might take $\exp(S)$ effort.
It is much more practical to treat~$C$ as a ``black box'' and apply \Cref{ineq:cheby-example}, expending just $O(S/\eps^2)$ effort to get a high-confidence estimate of~$\mu$ with ``error bar''~$\epsilon \sigma$.
This idea is the essence of the \emph{Monte Carlo Method}~\cite{HH64}.
Very surprisingly (at least, circa the mid-'90s), one can do quadratically better using a quantum computer!
As we show in this work, only $O(S/\eps)$ effort is needed to get the same guarantee.

To state our main result, let $\by$ be a discrete real random variable (whose values are encodable by bits on a digital computer).
We will formally discuss ``having the code'' for~$\by$ in \Cref{sec:prelims}, but for now suffice it to say it includes the following scenarios:
\paragraph{Scenarios for ``having the code'':}
\begin{enumerate} \label{list:code}
     \item \label{item:classical} Access to a classical randomized circuit (with no input) whose output is a draw from~$\by$.
     \item \label{item:general} More generally, access to a unitary quantum circuit (with some fixed input $\ket{0^k}$) such that, upon measuring its output and discarding some bits, we get a draw from~$\by$.
     (Note that a quantum circuit with intermediate measurements can be transformed to this form.) 
     \item \label{item:uniform} Less generally, access to a unitary quantum circuit that produces $\frac{1}{\sqrt{N}}\sum_{j =1}^N \ket{j}\ket{y_j}$, and $\by$~is defined to be the uniform distribution on the multiset of reals $\{y_1, y_2, \dots, y_N\}$.
     Grover's algorithm works in this model.
\end{enumerate}

In this work, we show the following theorem:
\begin{theorem} \label{thm:main0}
    There is a computationally efficient
    quantum algorithm with the following properties:
    Given ``the code'' for a random variable~$\by$, the algorithm uses $O(n)$ samples\footnote{In this introduction, we will say that an algorithm uses $q$~``samples'' from $\by$ to mean that it uses the code for~$\by$ at most~$q$ times.}
    and outputs an estimate~$\widehat{\bmu}$ such that
    \begin{equation}
        \Pr[\abs{\widehat{\bmu} - \mu} > \sigma/n] \leq 1/3,
    \end{equation}
    where $\mu = \E[\by]$ and $\sigma = \stddev[\by]$.
    (By repeating the algorithm $O(\log 1/\delta)$ times and taking the median, one can reduce the ``$1/3$'' to any $\delta > 0$.)
\end{theorem}

\begin{table}\centering
    \renewcommand{\arraystretch}{1.3}
    \begin{tabular}{@{}llcll@{}}\toprule
    \textbf{Model} & \textbf{Assumption on \emph{y}} & \textbf{Additive error} & \textbf{Samples} & \textbf{Reference} \\ \midrule
    Uniform & Bernoulli, $\mu = 0$ or $\mu = 1/n^2$ & (distinguishes${}^*$) & $O(n)$ & \cite{Gro96}\\
    Uniform & $[0,1]$-bounded & $1/n$ & $O(n \cdot \polylog\;n)$ & \cite{Gro98}\\
    General & Bernoulli & $\phantom {---} \sigma/n \ {}^{\dagger}$ \phantom{${}^{\dagger}$\ }& $O(n)$ & \cite{BHT98}\\
    General${}^\ddagger$ & $[0,1]$-bounded  & $\sqrt{\mu}/n$ & $O(n)$ & \cite{Ter99}\\
    Uniform & $\max\{\sigma, \abs{\mu}\} \leq \sigma_{\text{bound}}$ known & $\sigma_{\text{bound}} / n$ & $O(n \log^{3/2} n  \log \log n)$ & \cite{Hei02}\\
    General & $\sigma \leq \sigma_{\text{bound}}$ known & $\sigma_{\text{bound}} / n$ & $O(n \log^{3/2} n  \log \log n)$  & \cite{Mon15}\\
    General & (none) & $\sigma/n$ & $O(n \log^{3/2} n  \log \log n)$  & \cite{Ham21}\\
    \textbf{General} & (none) & $\boldsymbol{\sigma}\mathbf{/}\boldsymbol{n}$ & $\boldsymbol{O(n)}$  & \textbf{Us}\\
    \bottomrule
    \end{tabular}
    \caption{Prior quantum algorithms for  mean estimation.}
    {\small
    ${}^*$ Usually equivalently stated as using $O(\sqrt{N})$ queries to distinguish $\mu = 0$ from $\mu = 1/N$.\\
    ${}^\dagger$ Rather than just $\sigma/n = \sqrt{\mu(1-\mu)}/n$, this result is usually stated with an extra additive~$1/n^2$ error (cf.~\cite{BHMT00}).  But note that $1/n^2 \leq O(\sigma/n)$ unless $\sigma \ll 1/n$. Supposing $\sigma \ll 1/n$, we have $\mu \ll 1/n^2$ (or $1-\mu \ll 1/n^2$, but the reasoning will be similar), and inspecting the algorithm shows that it will output the estimate~$0$ (with high probability).  But an estimate of~$0$ \emph{is} within additive error~$\sigma/n$ of $\mu \sim \sigma^2$ when $\mu \ll 1/n^2$.\\
    ${}^{\ddagger}$ Terhal stated her result for the Uniform model, but it is easy to see it also works in the General model, as it is a direct reduction to the General Bernoulli result of~\cite{BHT98}.  The additive~$1/n^2$ appearing in her statement can be deleted for this reason, too.
    }
    \label{table:prior}
\end{table}

\begin{remark}  \label{rem:gatecxty}
    Regarding computational efficiency,
    in \Cref{sec:gates} we show that if the code for $\by$ is a circuit of gate complexity~$S$, then the gate complexity of our algorithm in \Cref{thm:main0} is~$O(nS)$.\footnote{Except in the rather specific and unlikely case of $\Omega(\log n) \leq S < o(\log(n) \cdot (\log \log n)^2)$, in which case there is an extra factor of at most $(\log \log n)^2$.}
\end{remark}

\Cref{thm:main0} is known to be optimal (up to constant factors, see e.g.~\cite[Thm.~4.6.2]{Ham21}), and it caps a long sequence of works that obtain similar results but with more assumptions and/or weaker parameters; see \Cref{table:prior}.  (The ``Uniform'' model in the table's first column refers to the model in \Cref{item:uniform} above.)

\vspace{5em}

\subsection{Methods}\label{sec:methods}

The centerpiece of \Cref{thm:main0} is \Cref{thm:maintask} below:
\begin{theorem} \label{thm:maintask}
    There is a computationally efficient 
    quantum algorithm that solves the following task:\\

    \textbf{Main Task.}  Given a parameter $\eps > 0$ and ``the code'' for a random variable~$\by$ promised to satisfy $\E[\by^2] \leq 1$,
    use $O(1/\eps)$ samples and distinguish (with confidence at least~$2/3$) between the cases (i)~$\abs{\mu} \leq \eps/2$ and (ii)~$\eps \leq \abs{\mu} \leq 2\eps$, where $\mu = \E[\by]$.\end{theorem}

One thing to notice is that \Cref{thm:maintask} directly implies Grover's algorithm~\cite{Gro96} (in its distinguishing form):\footnote{It's also not hard to show our routine can be used to \emph{find} a unique marked item.}
Given~$N$ equally likely items with either zero or one of them being ``marked'', we can form the random variable~$\by$ that is~$0$ on unmarked items and~$\sqrt{N}$ on a marked item.
Then we always have $\E[\by^2] \leq 1$, and either $\abs{\mu} = 0$ or $\abs{\mu} = 1/\sqrt{N}$ depending on whether there is a marked item.
Thus we can use~$\eps = 1/\sqrt{N}$ in \Cref{thm:maintask}.

In fact,  our algorithm for \Cref{thm:maintask} essentially \emph{is} Grover's algorithm --- but with complex phases!
Recall that Grover's algorithm is composed of a product of two unitaries.
The first unitary in Grover's algorithm, the ``diffusion operator'', is a reflection about the uniform superposition over all~$N$ items.
One can think of this state as encoding the uniform distribution over~$N$ equally likely items (but not which are marked/unmarked); if $N = 2^k$, it is the output of $H^{\otimes n}\ket{0^k}$.
In our algorithm, the diffusion operator will be similar; in the \Cref{item:general} scenario from above where the code is a quantum circuit~$\calC$, it would be reflection through the code's output $\calC \ket{0^k}$.

The second unitary in Grover's algorithm, the ``phase oracle'', is a diagonal unitary that encodes which items are marked: since in the Grover setup the random variable~$\by$ takes only two distinct values (marked or unmarked), these are mapped to the two phases $-1$~and~$+1$.
In our algorithm, the phase oracle will remain a diagonal unitary, inserting a phase based on the outcome~$y$ of the random variable.
But since $y$ may now be \emph{any} value from the real line, we need a map from an arbitrary real number to a phase. We map these values to general complex phases by associating $y \in \R$ to  $e^{\i \alpha} \in \C$, where $\alpha = -2\arctan y$, which maps $\R$ to $(-\pi,+\pi)$ (conceptually, this is the phase that rotates $1+\i y$ to $1 - \i y$). 
Notice that if~$\by$ just takes on the two values $0$~and~$\sqrt{N}$, the associated phases are $+1$ and~``almost~$-1$''.

\begin{figure}\label{fig:arctan}
    \centering
    \pgfplotsset{compat=1.17}
\begin{tikzpicture}
\begin{axis}[
     width=220pt,grid style={ultra thin},every axis plot post/.append style={thick},
     x tick label style={font=\tiny},y tick label style={font=\tiny},
     scale only axis,grid=major,axis lines=middle,
     xlabel={$y$},
     ylabel={$\alpha(y)$},
     xmin=-20,
     xmax=20,
     domain=-20:21,
     ymin=-4,
     ymax=4,
     xtick={-15,-10,...,15},
     ytick={-3, -2,...,3},
     restrict y to domain=-20:20,
     legend style={at={(0.5,-0.05)},anchor=north,nodes={right}},
 ]
 \addplot[mark=none,color=blue, samples=500]{-2*rad(atan(x))};
 \addlegendentry{$\alpha(y)=-2\arctan y$};
 \addplot [mark=none, black, dashed] {3.14159};
 \addplot [mark=none, black, dashed] {-3.14159};
 \addlegendentry{$\pm \pi$}
 \end{axis}
\end{tikzpicture}
~~~~
\begin{tikzpicture}[scale=3]

\draw[->] (-1.15,0) -- (1.15,0) coordinate (x axis);
    \draw[->] (0,-1.15) -- (0,1.15) coordinate (y axis);
    
\node [right]at (1.25,0)(x){$1$};    
    \node [above]at (0,1.25){$\i$};
    
\draw (0,0) circle [
                        radius=1cm
                        ];
                        
\path [name path=upward line] (1,0) -- (1,1);
    \path [name path=sloped line] (0,0) -- (30:1.5cm);
    \draw [name intersections={of=upward line and sloped line, by=t}]
            [] 
            (1,0) -- node [right=1pt]{$\displaystyle y $} (t);
    \draw (0,0) -- (t);
    
\path [name path=upward line] (1,0) -- (1,-1);
    \path [name path=sloped line] (0,0) -- (-30:1.5cm);
    \draw [name intersections={of=upward line and sloped line, by=r}]
            [] 
            (1,0) -- node [right=1pt]{$y$} (r);
    \draw (1,0) -- (r);
    
    \node at (1,0.7){$1+\i y$};
    \node at (1,-0.7){$1-\i y$};
    
\draw (r) coordinate (A)-- 
            (0,0) coordinate (B)-- 
                (t) coordinate (C)
    pic [
            draw,
            <-,
            blue,
            very thick,
            "\raisebox{11pt}{$\alpha$}",
            angle radius=11mm,
            angle eccentricity=1.3
            ] 
            {angle};
    \node at (0,-1.4){};
\end{tikzpicture}
\caption{Representations of the function $\alpha(y) = -2\arctan y$.}
\end{figure}
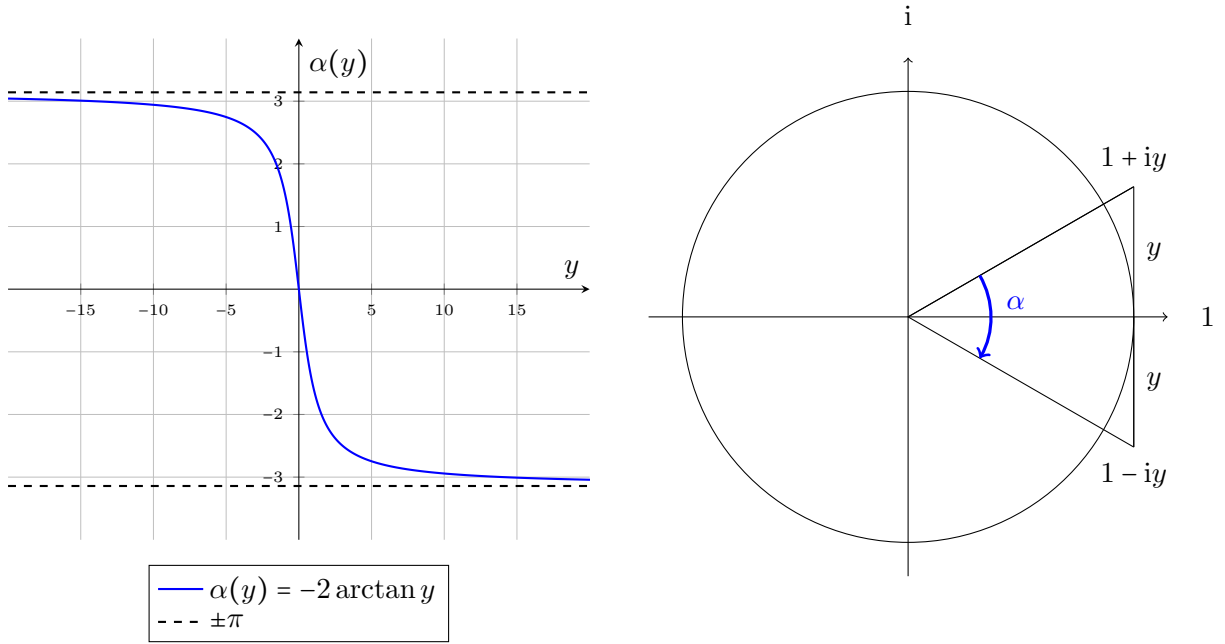

Grover's algorithm repeatedly alternates the diffusion and phase operators.  Our algorithm for the Main Task does the same, but the resulting intermediate states will contain general complex amplitudes.  (Nevertheless, the diffusion operator still acts by reflects a list of \emph{complex} amplitudes through their mean.)
This is not the first time using complex phases in Grover's has been suggested~\cite{Gro05}, but our analysis would seem to be new.  
The main challenge is to analyze the eigenvectors and eigenvalues of the unitary~$\RR$ obtained by composing the diffusion operator and the phased implementation of~$\by$.
Unlike in the Grover case, it is the composition of a reflection and a \emph{general} unitary operator, so $\RR$ is no longer essentially two-dimensional; it is fundamentally high-dimensional.

Nevertheless, we are able to show that when $\E[\by^2] \leq 1$, the natural ``starting state'' (namely, the output of the code) is mostly supported on eigenvectors of~$\RR$ with eigenvalue near $e^{\i \cdot 2|\mu|}$.
Given this analysis, we can solve the Main Task in \Cref{thm:maintask} rather easily.  
One way to finish is an immediate appeal to Quantum Phase Estimation~\cite{Kit95}. 
Alternatively, it's not too hard to show that more elementary strategies can work (provided one adjusts the constant factors in the Main Task's statement): Our analysis implies that form $T = \Theta(1/\eps)$, the starting state is close to an eigenvector of~$\RR^{T}$ with eigenvalue either close to~$+1$ (when $\abs{\mu} \ll \eps$) or close to~$-1$ (when $\abs{\mu} \approx \eps$). Then these cases can be distinguished with the simple Hadamard test.

\paragraph{Going from \Cref{thm:maintask} to \Cref{thm:main0}.}
A significant utility of our \Cref{thm:maintask} is that it applies to \emph{any} random variable with $\E[\by^2] \leq 1$, not just $\{0,1\}$-valued or even $[0,1]$-valued random variables.
As we will show in \Cref{sec:consequences}, this makes it very easy to compose with simple classical reductions.
For example, classical binary search lets us upgrade \Cref{thm:maintask} to \emph{estimate}~$\mu$ to additive $\pm \eps$ with $O(1/\eps)$ samples.
Then, with a standard classical halving trick we can achieve the Approximate Counting / Amplitude Estimation results of~\cite{BHT98,BHMT00}, as well as the~\cite{BHT98,Ter99} results from \Cref{table:prior}.
Indeed, the only extra quantum technique we use to reach our final \Cref{thm:maintask} is the recent Quantile Finding algorithm of Hamoudi~\cite{Ham21} (which itself is essentially Grover's algorithm together with classical reductions).
Thus all of these algorithms, up to and including \Cref{thm:main0},  can be obtained using nothing more than classical reductions and our elementary new quantum routine for \Cref{thm:maintask}.

\subsection{Applications of mean estimation} 

Mean estimation is used throughout algorithmic theory, and the quadratic speedup afforded by quantum computers (shown precisely herein) has many applications; see, e.g.,~\cite{Ham21} for an excellent survey.
We mention here some basic examples/applications.

\paragraph{Bernoulli random variables, and Circuit-SAT.}
Recall that a Bernoulli (i.e., $\{0,1\}$-valued) random variable~$\by$ which is~$1$ with probability~$\mu$ has mean~$\mu$ and standard deviation $\sigma = \sqrt{\mu(1-\mu)} \leq \sqrt{\mu}$.
Classically (or ``without the code''), $O(n)$~samples lets one estimate~$\mu$ to within additive error~$\sqrt{\mu}/\sqrt{n}$; so $O(1/\eps^2)$ samples always suffice for additive error~$\eps$, but also $O(1/\eps)$ samples suffice to distinguish $\mu \leq \eps/2$ from $\mu \geq \eps$.
Our \Cref{thm:main0} (and also the much earlier work on amplitude estimation~\cite{BHT98}) implies that ``with the code'', a quantum algorithm can improve these bounds to $n = O(1/\eps)$ and $n = O(1/\sqrt{\eps})$, respectively.

The simplest case of this corresponds to Grover's algorithm~\cite{Gro96} for the Unique-Circuit-SAT problem.
Suppose~$C$ is a classical $m$-input, $1$-output Boolean circuit with~$S$ gates, promised to be either unsatisfiable, or with a unique satisfying assignment.
Replacing its inputs by coin-flip gates, we thereby obtain classical ``code'' (in the model from \Cref{item:classical} above) for a Bernoulli random variable~$\by$ that either has mean $\mu = 0$ or mean $\mu = 2^{-m}$.
Grover's special case of \Cref{thm:main0} then shows that $n = O(1/\sqrt{2^{-m}}) = O(\sqrt{2}^m)$ samples --- and hence $O(\sqrt{2}^m \cdot S)$ quantum circuit complexity --- suffices to decide Unique-Circuit-SAT.

\paragraph{Distinguishing classical probability distributions.}
This application illustrates the importance of considering \emph{non}-bounded random variables, and achieving an additive guarantee that involves the standard deviation~$\sigma$.

Suppose $q$ and $r$ are fixed, known probability distributions on~$[D]$, and that an algorithm has access to samples from an unknown probability distribution~$p$ on~$[D]$, promised to be either $q$~or~$r$.
It is well known that the sample complexity needed to distinguish $p = q$ from $p = r$ (with error probability at most~$1/3$, say) is $\Theta(1/H(q,r)^2)$, where $H(q,r)^2 = \sum_{i=1}^D (\sqrt{q_i} - \sqrt{r_i})^2$ denotes the squared \emph{Hellinger distance} between~$q$ and~$r$.\footnote{Recall that $\tfrac12 H(q,r)^2 = 1-\mathit{BC}(q,r)$, where $\mathit{BC}(q,r) = \sum_{i=1}^D \sqrt{q_i}\sqrt{r_i}$ is the \emph{Bhattacharyya coefficient}, which clearly tensorizes: $\mathit{BC}(q^{\otimes n}, r^{\otimes n}) = \mathit{BC}(q,r)^n$.  Then the lower bound follows from the known total variation lower bound $\mathit{TV}(q,r) \geq \tfrac12 H(q,r)^2$.  For the upper bound, read on. \label{foot:BC}}
A recent work of Belovs~\cite{Bel19} shows that if one has ``the code'' for~$p$, there is a quantum algorithm that can distinguish $p = q$ from $p = r$ using just $\Theta(1/H(q,r))$ samples.

Here we show how Belovs's result is simply recovered from our \Cref{thm:main0}.
Writing $H = H(q,r)$, one way to distinguish $q$ and~$r$ is through mean estimation on the random variable $y : [D] \to \R$ defined by\footnote{The reader may verify that defining $0/0 = 0$ takes care of edge cases in what follows.}
\begin{equation}
    y(i) = \frac{\sqrt{q_i} - \sqrt{r_i}}{\sqrt{q_i} + \sqrt{r_i}} = \frac{(\sqrt{q_i} - \sqrt{r_i})^2}{q_i - r_i}.
\end{equation}
Writing $\mu_q = \E_q[\by]$, $\sigma^2_q = \Var_q[\by]$ and $\mu_r$, $\sigma_r$ analogously, observe that we have
\begin{equation}
    \mu_q - \mu_r = \sum_{i=1}^D (q_i - r_i) y(i) = \sum_{i=1}^D (\sqrt{q_i} - \sqrt{r_i})^2 = H^2
\end{equation}
and
\begin{equation}
    \sigma_q^2 + \sigma_r^2 \leq {\E}_q[\by^2] + {\E}_r[\by^2]  = \sum_{i=1}^D (q_i + r_i) \frac{(\sqrt{q_i} - \sqrt{r_i})^2}{(\sqrt{q_i} + \sqrt{r_i})^2} \leq H^2,
\end{equation}
where the inequality here uses $(q_i + r_i)/(\sqrt{q_i} + \sqrt{r_i})^2 \leq 1$.
From this we see that $p = q$ vs.\ $p = r$ can be distinguished by estimating the mean $\E_p[\by]$ to additive accuracy~$H^2/2$; and moreover,  that $\sigma_p \leq H$.
Classically (or without ``the code''), we need to ensure $\sigma_p/\sqrt{n}  \leq H/\sqrt{n} < H^2/2$, and can only say that  $n = O(1/H^2)$ samples suffice.
But \emph{with} the code, our quantum algorithm from \Cref{thm:main0} shows that only $H/n < H^2/2$ is needed; i.e., $n = O(1/H)$ samples suffice, matching Belovs's result.
(Indeed, careful inspection of Belovs's work shows that his algorithm can be thought of as performing mean estimation/distinguishing on~$\by$, additionally relying on the fact that~$H$ is a \emph{known} upper bound for~$\sigma$.)

\paragraph{Instance-dependent algorithms for multiplicative mean estimation.}
Consider the \cite{BHT98} entry of \Cref{table:prior}; it says that with a fixed budget of $O(n)$ samples, one can estimate the mean~$\mu$ of a Bernoulli random variable~$\by$ to additive error~$\sigma/n \leq \sqrt{\mu}/n$.
With this guarantee, one cannot even distinguish mean~$\mu$ from mean~$0$ unless $n \geq 1/\sqrt{\mu}$.
On the other hand, provided $n \geq 2/\sqrt{\mu}$, say, the estimate is accurate to within a \emph{multiplicative} factor of~$2$.
If this is one's only goal, one might wish for an algorithm that \emph{stops early} --- after only $O(1/\sqrt{\mu})$ samples --- obtaining a factor-$2$ approximation of~$\mu$ \emph{despite not knowing~$\mu$ a priori}.
This is the idea of \emph{sequential analysis}, from statistics.
Note that once~$\mu$ is known to a factor of~$2$, one can use the nonadaptive~\cite{BHT98} result to get a refined factor-$(1+\eps)$ approximation using $n = O(1/(\eps \sqrt{\mu}))$ samples.

This sort of ``instance-dependent'' guarantee was provided even for $[0,1]$-bounded random variables in the work~\cite{BHMT00}.
Precisely, their algorithm takes a parameter $\eps > 0$ and estimates the mean~$\mu$ of a $[0,1]$-bounded random variable to a factor of~$1+\eps$ using $O(1/(\eps \sqrt{\mu}))$ samples.

Following the theme of \Cref{table:prior}, one might wish to improve this result to take into account the standard deviation~$\sigma$.
Factor-$(1+\eps)$ approximation corresponds to additive error $\eps \mu$, and equating this with~$\sigma/n$ suggests that an improved instance-dependent sample complexity of $O(\sigma/(\eps \mu))$ might be possible.
Indeed, if a constant-factor upper bound on the ``coefficient of variation'' $|\sigma/\mu|$ happens to be known, this is immediate.
However if no prior assumptions are made, Hamoudi~\cite{Ham21} observes that a lower bound of Nayak~\cite{Nay99} shows that up to constant factors, no bound better than $O(\max\{\sigma/(\eps \mu), 1/\sqrt{\eps \mu}\} )$ is possible, even for Bernoulli random variables.
On the other hand, Hamoudi also uses the classical sequential analysis methods of~\cite{DKLR00} to show that the preceding bound \emph{can} be obtained, up to polylog factors, for $[0,1]$-bounded random variables.
The technique is a direct reduction to his instance-\emph{independent} result from \Cref{table:prior}.
We may apply the same reduction using our improved result, thereby obtaining:
\begin{theorem}
    In the setting of \Cref{thm:main0}, if $\by$ is $[0,1]$-bounded, there is an algorithm that, given $\eps > 0$, has the following behavior except with probability at most~$1/3$:  It obtains $O(\max\{\sigma/(\eps \mu), 1/\sqrt{\eps \mu}\})$  samples (dependent on the unknown~$\mu$), and outputs an estimate $\widehat{\bmu}$ such that
    \begin{equation}
        (1-\eps)\mu \leq \widehat{\bmu} \leq (1+\eps)\mu.
    \end{equation}
\end{theorem}

\paragraph{Algorithms for finance.} As another application, we briefly describe some real-world use cases of the Monte Carlo method in finance. We give only an overview of the financial terms used; for more information, we refer readers to a survey on quantum algorithms in finance~\cite{OML19,BvDJKP20,HGLGSSPA22} or a textbook on the mathematics of finance~\cite{Lue14}. Using quantum computers to solve the financial problems described below has been studied in some detail in prior work; see, e.g.,~\cite{RGB18,WE19,SESZISW20,CKMSWZ21}.

In finance, a derivative is a contract that derives its value from some underlying variable, such as the price of a specific stock. A ``European call option'' is a simple example of a derivative. Let us fix an underlying stock whose price on day $i$ is denoted $y_i$. A European call option with a strike price of $K$ and a maturity date of $T$ days is a contract that on day $T$ rewards the contract holder with $\max\{0,y_T-K\}$. In words, if the stock price on day $T$ is above $K$, the contract holder is rewarded with the difference, and otherwise receives nothing. Now if we have a probabilistic model for the daily price movement of the stock, we can infer a probability distribution over the possible values $y_T$. To determine a fair price for this call option (under our model), we need to compute the expected value of $y_T$ under this probability distribution.
In practice the probability distribution is efficiently sampleable, and the computational bottleneck is to compute this expected value, which is clearly a Monte Carlo mean estimation task. For example, a very simple model might be that the stock price increases by a factor of $f_0$ with probability $p$ and decreases by a factor of $f_1$ with probability $1-p$. A more commonly used model is to assume the stock price follows geometric Brownian motion, as in the Black--Scholes--Merton model~\cite{BS73,Mer73}. Under this model, European call options can actually be priced analytically, and we don't need to use Monte Carlo methods. But more generally, an option may depend on more than one underlying asset, and the payoff function can be be a complicated function of the entire history of stock prices $y_0,\ldots,y_T$. In such cases analytical solutions may not exist, but the Monte Carlo method works just fine as long as the distribution is efficiently sampleable and the payoff function is efficiently computable.

Note also that there is no generic reason why the derivative price should be in a known bounded interval or have known standard deviation bound.
Indeed, in the analysis from~\cite{CKMSWZ21} on pricing autocallable and TARF derivatives, significant gate complexity arose due to the errors incurred by artificially truncating prices to bounded intervals.  This suggests that our methods, which don't require any such bounds, might be helpful.

Other examples from finance where Monte Carlo methods are used in practice include the computation of Value at Risk (VaR) and Conditional Value at Risk (CVaR); these give more examples where the techniques of this paper can be used to give a quadratic speedup using a quantum computer.

\paragraph{More applications in TCS.} There are innumerable additional applications of Monte Carlo mean estimation throughout theoretical computer science --- simulated annealing algorithms, approximation of partition functions,  MCMC approximate counting algorithms, subgraph count estimation algorithms, data stream estimation algorithms, etc.; see, e.g.,~\cite{Mon15,Ham21} for some illustrations.

\section{Preliminaries}     \label{sec:prelims}

In this section we formally define random variables, and what it means to ``have the code'' for them.
In short, we use the same model as Montanaro~\cite{Mon15} --- essentially, \Cref{item:general} in the scenarios from \Cref{list:code}.
(The reader may also refer to the thesis of Hamoudi~\cite{Ham21}, where the model is discussed in careful detail.)

\subsection{Probability distributions}
Before defining random variables, we discuss probability distributions.
As our random variables will be implemented by finite circuits, it suffices to discuss finite probability distributions.
\begin{definition}[Finite probability space]
    A finite probability space is a pair $(\Omega,p)$ where $\Omega$ is a finite set of ``outcomes'' and $p:\Omega \to \R$ is a probability distribution, satisfying  $p(\omega) \geq 0$ for all $\omega \in \Omega$, and $\sum_{\omega \in \Omega} p(\omega)=1$.
\end{definition}
As suggested in the scenarios of \Cref{list:code}, there are several ways a probability distribution may be implemented with a classical or quantum circuit.
(See \cite[Sec.~3]{Bel19} for a somewhat detailed discussion/comparison.)
We will prefer the most general one there, \Cref{item:general}, where a draw from the distribution is obtained by measuring the output of a quantum circuit~$\U$ and possibly discarding some of the result.
This is a very natural model, though it does not seem to have a common and succinct name; since~$\U$ can be used to produce ``synthetic data'', we will term it a ``synthesizer'':
\begin{definition}[Synthesizer]\label{def:synthesizer}
    Let $p$ be a probability distribution on~$\Omega$.
    A synthesizer for~$p$ is any unitary circuit~$\U$ that performs the map
    \begin{equation}    \label{eqn:synthoutput}
            \U \ketzero = \sum_{\omega \in \Omega} \sqrt{p(\omega)}\ket{\omega}\ket{\garbage_\omega},
    \end{equation}
    where: $\ketzero$ is any easy-to-prepare fixed state (for definiteness, say $\ketzero = \ket{0^k}$ where $k$ is the number of input qubits to~$\U$); $\ket{\omega}$ is a normalized vector representing~$\omega$; and, $\ket{\garbage_\omega}$ is any normalized ``garbage vector''.
    In a typical implementation we would have the space $\Omega = [D]$ (for $D$ a power of~$2$), with $\ket{\omega}$ being the $\log_2 D$-qubit representation of~$\omega$ in the computational basis.
\end{definition}
Observe that in the typical scenario, if we produce $\U \ketzero$ and measure the first register (ignoring/discarding the second), we indeed obtain the outcome~$\omega$ with probability~$p(\omega)$.
\begin{example}
    For $\Omega = \{0, 1, 2, \dots, 2^n-1\}$, the uniform distribution~$p$ on~$\Omega$ has, as a synthesizer, the circuit that consists of applying a Hadamard gate to each qubit of $\ketzero = \ket{0^n}$.
\end{example}
\begin{example}
    Continuing the previous example, suppose that as in Grover's algorithm, an additional circuit ``marks'' the outcomes $M \subseteq \Omega$ by producing
    \begin{equation}
        \sum_{\ell=0}^{2^n-1} \ket{y_\ell}\ket{\ell} , \quad \text{where } y_\ell = \begin{cases}
            1 & \text{if $\ell \in M$,} \\
            0 & \text{if $\ell \not \in M$.}
        \end{cases}
    \end{equation}
    (Here contrary to common convention, we put the ``flag register'' on the left rather than the right.)
    The above state may be rewritten as
    \begin{equation}    \label{eqn:grovsyn}
        \sqrt{\tfrac{2^n - |M|}{2^n}} \ket{0} \ket{\garbage_0}     + \sqrt{\tfrac{|M|}{2^n}} \ket{1} \ket{\garbage_1} 
    \end{equation}
    with ``garbage vectors''
    \begin{equation}
        \ket{\garbage_0} = \tfrac{1}{\sqrt{2^n - |M|}}\sum_{\ell \not \in M} \ket{\ell}, \qquad 
        \ket{\garbage_1} = \tfrac{1}{\sqrt{|M|}}\sum_{\ell \in M} \ket{\ell}.
    \end{equation}  
    In this way, the composite algorithm producing the state in \Cref{eqn:grovsyn} may be seen as a synthesizer for the two-outcome Bernoulli probability with parameter $p = \frac{|M|}{2^n}$.
\end{example}

\begin{definition}[Having the code for a distribution]
        When we use the phrase ``having the code'' for a distribution~$p$, we refer only to having black-box access to controlled-$\U$ and controlled-$\U^\dagger$, where $\U$~is a synthesizer for~$p$.\footnote{Grover-type algorithms use~$\U^\dagger$, but we point out it is also very reasonable to require controlled-$\U$.  Without it, there would be no way to implement a synthesizer for a simple modification of~$p$ such as ``with probability $1/2$ draw from~$p$, with probability $1/2$ output~$0$''.}
        Of course, if we have white-box access to the circuit~$\U$, we can easily produce circuits for controlled-$\U$ and controlled-$\U^\dagger$.
\end{definition}

\begin{remark}
    Allowing ``garbage vectors'' in \Cref{eqn:synthoutput} is crucial to obtain an acceptable level of generality.
    Insisting that a synthesizer produce a ``coherent'' (garbage-free) version of \Cref{eqn:synthoutput} (as in the ``uniform'' model \Cref{item:uniform} of \Cref{list:code}) would be very limiting.  Indeed, efficient coherent synthesizers are typically impossible even when there is an efficient \emph{classical} sampling algorithm for the distribution.\footnote{For example, given a graph $G$, it is easy to uniformly sample from all automorphisms (vertex-labelings) of the graph. But if we could efficiently synthesize a quantum state that is the uniform superposition over all automorphisms, then we would be able to solve  Graph Isomorphism efficiently: We would simply create the state for~$G$, the state for~$H$, and check if they are the same or orthogonal.}
\end{remark}

\paragraph{Classical algorithms.}
The above synthesizer definition also covers the case of having a classical circuit with coin-flip gates that generates draws from~$p$ (as in the \Cref{item:classical} scenario).
Such a circuit can be converted into a deterministic circuit that accepts $r$~random bits as input. Then the deterministic circuit can be made reversible (with constant-factor overhead) using classical Toffoli and NOT gates~\cite{Ben73}.
The resulting  circuit will now accept some~$a$ ancillary input bits set to~$0$, and output the original output, along with some additional garbage bits. Next, we make a quantum circuit with the same behavior on classical basis states by replacing all the classical Toffoli and NOT gates with quantum Toffoli and quantum NOT gates.
Finally, the original $r$~random input bits can be replaced by~$r$ Hadamard gates with with~$\ket{0}$ inputs.  The result is a synthesizer circuit~$\U$ for the probability distribution, with gate complexity only a constant factor larger than that of the original classical circuit.\footnote{To see this explicitly worked out, see Appendix~A of the arXiv version of \cite{WSK+21}}

\subsection{Random variables}
Now we come to the main object of study in this paper, real random variables.
Formally, a (discrete, real) random variable is just a real-valued function on a probability space:
\begin{definition}[Random variable]
    Given a finite probability space $(\Omega,p)$, a random variable~$\by$ is defined by a function $y:\Omega \to \R$.
\end{definition}

\begin{definition}[Moments of a random variable]
    For a random variable $\by$, we define the following:
    \begin{align}
        \mu &\coloneqq \E[\by] = \sum_{\omega \in \Omega} p(\omega)y(\omega) \tag{mean or expected value}\\
        s^2 &\coloneqq \E[\by^2] = \sum_{\omega \in \Omega} p(\omega) y(\omega)^2
        \tag{second moment}    \\
        \sigma^2 &\coloneqq \Var[\by] = \E[(\by-\mu)^2] = \E[\by^2]-\mu^2 = s^2 - \mu^2
        \tag{variance} \\
\sigma &= \stddev[\by] \tag{standard deviation}
    \end{align}
\end{definition}
\begin{remark}
Our notation for the raw second moment,~$s^2$, is not standard (as opposed to the standard notation~$\sigma^2$ for the variance, aka ``central second moment'').
\end{remark}

Regarding implementation of a random variable $y : \Omega \to \R$, we assume a standard quantum oracle:
\begin{definition}[Having the code for a random variable]
    When we use the phrase ``having the code'' for a random variable~$\by$ on probability space~$(\Omega,p)$, this refers to having a synthesizer for~$p$, as well as having access to controlled-$\calY$ and controlled-$\calY^\dagger$, where $\calY$ is any unitary circuit with the behavior
    \begin{equation}    \label{eqn:FF}
        \calY \ket{\omega}\ket{0^b}\ket{0^c} = \ket{\omega}\ket{y(\omega)}\ket{0^c}
    \end{equation}
    for all $\omega \in \Omega$.
    Here it is assumed the real range of~$\by$ is encoded using~$b$ bits (e.g., in fixed-point representation); the extra~$c$ bits are for ancillas.
\end{definition}
\begin{remark}
    As with classical implementations of probability distributions, given a classical circuit computing $y : \Omega \to \R$ (with appropriate input/output encoding), one can efficiently convert it to a quantum circuit~$\calY$ as above (taking care to uncompute garbage).
\end{remark}
\begin{remark}
    Some readers may find it overly fussy that we have insisted on the mathematical definition of random variables as functions on probability spaces.
    However, it will be very convenient in our work to think of them in this way.\footnote{Such readers may also recall, e.g., how much simpler it is to prove Linearity of Expectation from the definition $\E[\by] = \sum_{\omega} p(\omega) y(\omega)$ than from $\E[\by] = \sum_y \Pr[\by = y] x$.}

    Consider also the unfussy notion of a random variable~$\by$ being implemented by a quantum (or classical randomized) circuit~$\calC$, wherein measuring $\calC\ketzero$ (and discarding garbage) directly yields a draw from~$\by$.
    In this case, we can formally define $\Omega = \mathrm{range}(\by)$, define $p(\omega) = \Pr[\by = \omega]$, treat~$\calC$ as a synthesizer for~$p$, and formally take~$y : \Omega \to \R$ to be the identity map (so that $\calY = \Id$ and $b = c = 0$ in \Cref{eqn:FF}).
\end{remark}

\section{Establishing \Cref{thm:maintask} --- Grover with complex phases}   \label{sec:mainthing}
\newcommand{\sbd}{\tfrac{1}{16}}
In this section we establish \Cref{thm:maintask}. First, in \Cref{sec:alg-desc} we fully describe the algorithm, which involves setting up a unitary~$\calU$ and performing phase estimation with an initial state~$\ket{\bone}$.
The next \Cref{sec:qpe} gives some generic preliminaries on phase estimation.
Subsequently, we need to analyze the eigenvalues and eigenvectors of our particular~$\calU$ --- or at least the eigenspaces in which~$\ket{\bone}$ mostly resides.
We introduce some notation in \Cref{sec:states}; then in \Cref{sec:eigo} we derive the key eigenvalue inequalities for our analysis and do two things:
\begin{itemize}
    \item Show the inequalities easily imply that  Quantum Phase Estimation achieves \Cref{thm:maintask}, except with worse constants (which nevertheless would suffice to solve our overall Mean Estimation task).
    \item Show that a sharper analysis of the eigenvalue inequalities would lead to \Cref{thm:maintask} with its constants as stated.
\end{itemize}
Subsequently in \Cref{sec:details}, we give the sharper analysis of the eigenvalue inequalities.
In \Cref{sec:elementary}, we observe that using Quantum Phase Estimation as a black box is arguably overkill for our problem (though it makes the analysis succinct); we illustrate how one can instead complete the algorithm via measuring~$\ket{\bone}$ against~$\calU^T \ket{\bone}$ for $T = \Theta(1/\eps)$.
Finally, in \Cref{sec:rankone} we observe that, in a certain sense, all of the eigenvalues and eigenvectors of~$\calU$ can be described geometrically and somewhat simply. 
We found that this description did not seem to simplify any of our preceding analysis, though it may aid in intuition.

\subsection{Algorithm description} \label{sec:alg-desc}
Since our algorithm is essentially just Grover's algorithm with complex phases, it's easy to fully describe the algorithm and its complexity. Proving correctness of the algorithm will then be the goal of the subsequent subsections.

Let $([D], p)$ be a probability space implemented by synthesizer~$\U$, so
\begin{equation}    \label{eqn:p0}
    \U \ketzero = \sum_{\ell=1}^D \sqrt{p(\ell)} \ket{\ell} \ket{\garbage_\ell}
\end{equation}
as in \Cref{eqn:synthoutput}.  Here we have written $\ell \in [D]$ instead of $\omega \in \Omega$ to make the notation less laborious (but note that the $\ket{\ell}$ above, really $\ket{\omega}$, need not literally denote the $\ell$th standard basis vector). Let~$\by$ be a real random variable defined by $y : [D] \to \R$ and  computed by circuit~$\calY$ as in \Cref{eqn:FF}.
In this section we will write
\begin{equation}
    y_\ell \text{ instead of } y(\ell).
\end{equation}

We may now define the key unitary used by our algorithm that accomplishes the Main Task from \Cref{thm:maintask}. 
As in Grover's algorithm, it is composed of two parts:
\begin{equation}\label{eq:mainU}
    \calU = \Refl \cdot \RotF.
\end{equation}
\begin{definition}  \label{def:refl}
    The operator $\Refl$ (essentially the ``Grover diffusion operator'' vis-a-vis~$p$) is defined by
    \begin{equation}
        \Refl = \text{$\U$}(2\text{$\ketzero\!\brazero$} - \text{$\Id$})\text{$\U^\dagger$}.
    \end{equation}
\end{definition}
\begin{definition}
    The operator $\RotF$ (the ``phase oracle'') is defined by
    \begin{equation}
        \RotF \ket{\ell}\ket{\garbage_\ell} = e^{\i \alpha_\ell} \ket{\ell}\ket{\garbage_\ell} \quad \forall \ell \in [D],
    \end{equation}
    where the angles $\alpha_\ell \coloneqq -2 \arctan y_\ell$ are defined so that
    \begin{equation}    \label{eqn:keyrot}
        e^{\i \alpha_\ell}(1+\i y_\ell) = 1-\i y_\ell.
    \end{equation}
\end{definition}
\begin{remark}    \label{rem:uses}
    The operator $\Refl$ is evidently efficiently computable using two applications of ``the code'': one application of~$\U$ and one application~$\U^\dagger$.
    Additionally, as we will later use quantum phase estimation (or, at least, the Hadamard test), we will in fact require controlled-$\calU$, not just $\calU$ itself, and hence  will really need applications of controlled-$\U$ and controlled-$\U^\dagger$.

    The operator $\RotF$ is also efficiently computable using two applications of ``the code'': Given~$\calY$ as in \Cref{eqn:FF}, we adjoin~$\ket{0^b}$ and apply~$\calY$ to 
    get $\ket{\ell}\ket{\garbage_\ell}\ket{y_\ell}$.\footnote{Formally, we will also need to adjoin ancillas, but these will always be set to all-$\ket{0}$'s and restored to all-$\ket{0}$'s, and thus may be safely ignored. As is conventional, we will avoid further mention of them.}
    We may then employ a classical routine (the computational efficiency and precision of which are discussed in \Cref{sec:gates}) to compute~$\ket{\alpha_\ell}$ from~$\ket{y_\ell}$, multiply by the phase $e^{\i \alpha_\ell}$, uncompute with the help of~$\calY^\dagger$, and thus finally reach $e^{\i \alpha_\ell}\ket{\ell}{\ket{\garbage_\ell}}$.  
    Recall again that we will eventually use controlled-$\calU$, not just~$\calU$, and hence again we will really need one application each of controlled-$\calY$ and controlled-$\calY^\dagger$.

    Thus overall (controlled-)$\calU$ can be efficiently implemented with four uses of ``the code'' for~$\by$.
\end{remark}

The main claim in the later analysis is that the state $\U \ketzero$ has high overlap with the eigenvectors of $\calU$ of eigenphase approximately $2\abs{\mu}$. Then employing phase estimation with precision $\eps/6$, which requires $O(1/\eps)$ uses of~$\calU$, will allow us to distinguish the two ranges of~$\abs{\mu}$ and complete the proof of \Cref{thm:maintask}.

\subsection{Generic phase estimation setup} \label{sec:qpe}
In this subsection, let $\calU$ denote \emph{any} generic unitary operator on $\C^D$.
Suppose we perform phase estimation (or a simpler, Grover-like algorithm) with~$\calU$ and ``starting state''~$\ket{\sigma}$.
Then we will need to know about the eigenvalue(s) of~$\calU$ corresponding to the eigenvector(s) that~$\ket{\sigma}$ is close to.
Let us introduce some notation to facilitate this:
\begin{notation}    \label{not:eigen}
    Fix an eigendecomposition of $\calU$,
    \begin{equation}
        \calU = \sum_{j=1}^D e^{\i \theta_j} \ket{u_j}\!\bra{u_j},
    \end{equation}
    with $-\pi < \theta_j \leq \pi$.
    Given some $\ket{\sigma}$, we express it in $\calU$'s eigenbasis as
    \begin{equation}    \label{eqn:zeta}
        \ket{\sigma} = \sum_{j=1}^D \hat{\sigma}_j \ket{u_j}, \qquad \hat{\sigma}_j \coloneq \braket{u_j|\sigma}.
    \end{equation}
    When $\ket{\sigma}$ is a \emph{unit} vector,
    Pythagorus tells us the squared coefficients $|\hat{\sigma}_1|^2, \dots, |\hat{\sigma}_D|^2$ form a probability distribution on~$[D]$.
    In this case we will write $\bj \sim J_{\calU}(\ket{\sigma})$ to denote that~$\bj$ is drawn according to this probability distribution;
    we will also write $\btheta \sim \Theta_{\calU}(\ket{\sigma})$ to denote that $\btheta$ is the random angle formed by drawing $\bj \sim J_{\calU}(\ket{\sigma})$ and then setting $\btheta = \theta_{\bj}$.
\end{notation}
\begin{remark}  \label{rem:QPE}
    One could say the random variable $\btheta \sim \Theta_{\calU}(\ket{\sigma})$ is the output of ``Idealized Phase Estimation''; that is, phase estimation making no  error.
    Indeed, the \emph{actual} behavior of Quantum Phase Estimation~\cite{Kit95,CEMM98} when run with~$\calU$ and $\ket{\sigma}$ is that, after $O(\log(1/\delta)/\eps)$ applications of controlled-$\calU$, the output is a random variable~$\btheta'$ with the following property:
    
    \emph{There is a probabilistic coupling between $\btheta$ and $\btheta'$ under which $\Pr[|\btheta - \btheta'| > \eps] \leq \delta$.}
\end{remark}

\medskip

On the topic of closeness between $J_{\calU}(\cdot)$ distributions, the following fact relates the fidelity between two different starting states and the Hellinger distance between their associated~$J_{\calU}(\cdot)$'s:
\begin{proposition} \label{prop:bhatt}
    Given $\calU$ as in \Cref{not:eigen}, suppose $\ket{\sigma}, \ket{\tau} \in \C^d$ are unit vectors. Write $q_1, \dots, q_d$ (respectively, $r_1, \dots, r_d$) for the probabilities of $J_{\calU}(\ket{\sigma})$ (respectively $J_{\calU}(\ket{\tau})$).
    Then we have the following Bhattacharyya coefficient / Hellinger-squared bound:
    \begin{equation}
        \mathit{BC}(q,r) \geq \abs{\braket{\sigma|\tau}}; \qquad\text{in other words,}\qquad H(q,r)^2 \leq 2(1-\abs{\braket{\sigma|\tau}}).
    \end{equation}
\end{proposition}
\begin{proof}
    Writing $\lambda_j = \sqrt{r_j/q_j}$ (and taking $\lambda_j = 0$ if $q_j = 0$), we have
    \begin{align}    \label[ineq]{ineq:hqq}
        H(q,r)^2 = \sum_j q_j (1 - \lambda_j)^2 
= 2\parens*{1 - \sum_j \sqrt{q_j} \sqrt{r_j} }
        &= 2\parens*{1 - \sum_j \abs{\braket{u_j|\sigma}} \cdot \abs{\braket{u_j|\tau}} }
        \\
        &\leq 2 \parens*{1 - \abs*{\sum_j \braket{\sigma|u_j} \cdot \braket{u_j|\tau}} }
        = 2 \parens*{1 - \abs*{\braket{\sigma|\tau}}}.  \qedhere
    \end{align}
\end{proof}

When it comes to analyzing eigenvalues $e^{\i \theta}$ of~$\calU$, we will use the following quantity --- quaintly called the \emph{haversine} of angle~$\theta$ --- to measure how ``nontrivial'' rotation-by-$\theta$ is:
\begin{notation}
        For any $\theta \in \R$ we may write
        $\displaystyle
            \hav \theta \coloneqq \abs*{\frac{1-e^{\i \theta}}{2}}^2 = \frac{1-\cos\theta}{2} = \parens*{\sin \frac{\theta}{2}}^2  \in [0,1].
        $
\end{notation}
From \Cref{eqn:zeta} we  have $\parens*{\frac{\text{$\Id$}-\calU}{2}}^{\pm 1} \ket{\sigma} = \sum_j \hat{\sigma}_j \parens*{\frac{1-e^{\i\theta_j}}{2}}^{\pm 1} \ket{u_j}$, and thus the above fact implies:
\begin{proposition}     \label{prop:sin2-dilation}
    Let $\ket{\sigma}$ be a unit vector, and assume for $\btheta \sim \Theta_{\calU}(\ket{\sigma})$ that $\btheta$ is never~$0$.  Then
    \begin{equation}
        \norm*{ \parens*{\frac{\text{$\Id$} - \calU}{2}}^{\pm 1} \ket{\sigma}}^2 
        = \E_{\btheta \sim \Theta_{\calU}(\ket{\sigma})}\bracks*{(\hav \btheta)^{\pm 1}}.
    \end{equation}
\end{proposition}

\subsection{Quantum states corresponding to complex random variables} \label{sec:states}

We now return to our particular $\calU = \Refl \cdot \RotF$ as described in \Cref{eq:mainU}. Let us introduce some notation that allows us to conveniently talk about states on which this unitary acts.

\begin{notation}
    Let $z_1, \dots, z_D$ be any complex numbers.
    We may think of this list as defining a \emph{complex-valued} random variable~$\bz$ on~$([D],p)$.
    (To draw from~$\bz$, first choose $\bell \in [D]$ according to~$p$ and then set $\bz = z_{\bell}$.)
    Then we will also define the (not necessarily unit) vector
    \begin{equation}
        \ket{\bz} = \sum_{\ell=1}^D z_\ell \sqrt{p(\ell)} \ket{\ell} \ket{\garbage_\ell}.
    \end{equation}
    In particular, referring to \Cref{eqn:p0} we have
    \begin{equation}    \label{eqn:Uz}
        \U \ketzero = \ket{\bone},
    \end{equation}
    where $\bone$ denotes the random variable on~$([D],p)$ that is constantly~$1$.
\end{notation}
It is easy to compute the following:
\begin{fact}    \label{fact:ip}
    For complex-valued random variables $\bw,\bz$ on $([D],p)$ we have 
    \begin{equation}
        \braket{\bw|\bz} = {\E}_p[\overline{\bw} \bz] = \sum_{\ell=1}^D p(\ell) \ol{w}_\ell z_\ell.
    \end{equation}
    In particular, $\braket{\bone|\bz} = \E_p[\bz]$.
\end{fact}
\begin{remark}
    We will often consider \emph{non}-unit vectors~$\ket{\bz}$. 
    The vector $\ket{\bz}$ is only a properly normalized quantum state if $\E_p[|\bz|^2] = 1$.
    (For example, $\ket{\bone}$ \emph{is} a valid quantum state.)
\end{remark}
\begin{remark}
    Even if $\ket{\bw}, \ket{\bz}$ are unit vectors, and thus may be considered quantum states, one should \emph{not} consider them to be identical if they are equal up to a global phase.
    The reason is our algorithm will eventually introduce a control qubit (for phase estimation purposes), which will make global phases into relative phases.
\end{remark}

With this notation for states, we can now examine what the reflection operator does to a state. Recalling our new notation (particularly \Cref{eqn:Uz}), our reflection operator defined in \Cref{def:refl} is
\begin{equation}
    \Refl = \text{$\U$}(2\text{$\ketzero\!\brazero$} - \text{$\Id$})\text{$\U^\dagger$} = 2\ket{\bone}\!\bra{\bone} - \text{$\Id$}.
\end{equation}
    From \Cref{fact:ip} we see that
    $
        \Refl \ket{\bz} = \ket{\bz_{\mathrm{refl}}},
    $
    where
    \begin{equation}    \label{eqn:reff}
        \bz_{\mathrm{refl}} = 2{\E}_p[\bz] - \bz
    \end{equation}
    is the random variable in which  each~$z_\ell$ is replaced with its reflection through the ``barycenter''~$\E_p[\bz]$.

\subsection{Eigenvalue analysis}    \label{sec:eigo}
In this section we do the eigenvector/eigenvalue analysis of our operator $\calU = \Refl \cdot \RotF$, as a function of the random variable~$\by$.
We will use the notation
\begin{equation}
    \mu = {\E}_p[\by], \qquad s^2 = {\E}_p[\by^2].
\end{equation}
\noindent The key vectors for our analysis are the ``starting state''~$\ket{\bone}$, and the following vector:
\begin{equation}
    \ket{\bone + \bii \by} \coloneqq \sum_{\ell=1}^D (1+\i y_\ell)\sqrt{p(\ell)} \ket{\ell}\ket{\garbage_\ell}.
\end{equation}
Using \Cref{fact:ip}, we have:
\begin{fact}    \label{fact:ourip}
    $\displaystyle \braket{\bone + \bii \by|\bone + \bii \by} = {\E}_p[\bone^2 + \by^2] = 1+s^2$; and, $\braket{\bone|\bone + \bii \by} = 1 + \i \mu$, so $\displaystyle \abs*{\braket{\bone|\bone + \bii \by}} = \sqrt{1 + \mu^2} \geq 1$.
\end{fact}
From this we see that if~$s$ is small (as we will assume), then $\ket{\bone + \bii \by}$ is close to being a unit vector, and this unit vector is close to~$\ket{\bone}$.
Let us introduce a normalized version of the vector:
\begin{notation}
    We write $\ket{\bonef} = \frac{1}{\sqrt{1+s^2}}\ket{\bone + \bii \by}$,  a unit vector.
\end{notation}

Perhaps the key intuition behind the analysis is that if $\mu = 0$, then the vector $\ket{\bone + \bii \by}$ is \emph{fixed} by~$\calU$ (i.e., it is an eigenvector of eigenvalue~$1$).
The following proposition generalizes this fact:
\begin{proposition} \label{prop:whogoeswhere}
$\displaystyle \frac{\Id - \RR}{2} \ket{\bone + \bii \by} = \i \mu \ket{\bone}.$
\end{proposition}
\begin{proof}
    This follows from
    \begin{align}
        \RR \ket{\bone + \bii \by} &= \Refl \cdot \RotF \ket{\bone + \bii \by} \\
        &= \Refl \ket{\bone - \bii \by} \qquad \textrm{(using \Cref{eqn:keyrot})}\\
&= \ket{\bone + \bii(\by - \boldsymbol{2}\boldsymbol{\mu})} \qquad\textrm{\,(using \Cref{eqn:reff})}\\
        &= \ket{\bone + \bii \by}  - 2i\mu \ket{\bone}. 
        \qedhere
    \end{align}
\end{proof}

\begin{proposition} \label{prop:verse}
    Writing $\wt{\btheta} \sim \Theta_{\RR}(\ket{\bonef})$, we have
    $\E[\hav \wt{\btheta}] = \mu^2/(1+s^2)$.
\end{proposition}
\begin{proof}
    This is immediate by taking the squared-length of both sides in \Cref{prop:whogoeswhere} and then applying \Cref{prop:sin2-dilation} with exponent~$+1$.
\end{proof}

Somewhat peculiarly, for $\btheta \sim \Theta_{\calU}(\ket{\bone})$ we can also determine the expected \emph{reciprocal} of~$\hav \btheta$:
\begin{proposition} \label{prop:inverse}
    Writing $\btheta \sim \Theta_{\RR}(\ket{\bone})$, we have $\E[(\hav \btheta)^{-1}] = (1+s^2)/\mu^2$.

    \noindent (Technically, we must assume that $\btheta$ is never~$0$ and that $\mu \neq 0$.)
\end{proposition}
\begin{proof}
    Rearranging the statement of \Cref{prop:whogoeswhere} gives
    \begin{equation}
        \parens*{\frac{\Id - \calU}{2}}^{-1}\ket{\bone} = \frac{1}{\i \mu} \ket{\bone + \bii \by}.
    \end{equation}
    The proof is completed by taking the squared-length on both sides and then applying \Cref{prop:sin2-dilation} with exponent~$-1$.
\end{proof}

Recall that, assuming~$s$ is small, we have that $\ket{\bone}$ is close to~$\ket{\bonef}$, and hence~${\btheta}$ should be similar in distribution to~$\wt{\btheta}$.
The preceding two propositions therefore suggest that $\hav \btheta = \sin^2(\btheta/2)$ ought to concentrate around $\mu^2/(1+s^2)$; i.e., $|\btheta|$ ought to concentrate around $2|\mu|$ (again, when assuming~$s$ is small).
Indeed, we can use them to establish the following:
\begin{theorem}   \label{thm:multi}
    For certain constants $s_0, c_0, 1/C_0, \delta_0 > 0$, the following holds:  Provided $s \leq s_0$, for $\btheta \sim \Theta_{\calU}(\ket{\bone})$ we have
    \begin{equation} \label[ineq]{ineq:new}
        \Pr\Bigl[c_0 \cdot 2|\mu| \leq |\btheta| \leq C_0 \cdot 2|\mu|\Bigr] \geq 1-\delta_0.
    \end{equation}
    In particular (see \Cref{cor:done}), we may take $s_0 = \frac{1}{16}$, $c_0 = \frac45$, $C_0 = \frac54$, $\delta_0 = \frac29$.
\end{theorem}
With these specific ``in particular'' constants, we can complete the proof of \Cref{thm:maintask} almost immediately by using phase estimation. 
On the other hand, achieving these constants is slightly fiddly; hence, we defer this to \Cref{sec:details}. 
For now, we illustrate how \Cref{thm:multi} can be proven in a very simple way, allowing for worse constants.  
(As we note in \Cref{rem:constants0,rem:constants} and \Cref{sec:elementary}, these worse constants are still sufficient for giving elementary proofs of our main results \Cref{thm:main0,thm:maintask}.)
\begin{proof}[Proof of \Cref{thm:multi} with worse constants.]
    We establish the theorem with 
    \begin{equation}    \label{eqn:simple}
        s_0 = .001, \quad c_0 = .05, \quad C_0 = 50, \quad \delta_0 = .005.
    \end{equation}
    Beginning with a technicality, note that the conclusion of our theorem is continuous with respect to infinitesimally perturbing~$\calU$; thus we may assume without loss of generality that $\mu, \btheta \neq 0$ always holds.  
    Now applying Markov's inequality to \Cref{prop:inverse} we get
    \begin{equation}
        \Pr[(\hav \btheta)^{-1} > 350(1+s^2)/\mu^2] \leq \frac{1}{350} < .003.
    \end{equation}
    Assuming $s \leq .001$, we conclude that except with probability less than $.003$ we have
    \begin{equation}
        (\hav \btheta)^{-1} \leq 350(1+s^2)/\mu^2 
        ~~\implies~~ \sin^{-2}(\btheta/2) \leq 400/\mu^2 ~~\implies~~ (\abs{\mu}/20)^2 \leq \sin^2(\btheta/2) \leq (\btheta/2)^2.
    \end{equation}
    Hence
    \begin{equation} \label[ineq]{ineq:a1}
        \Pr[\abs{\btheta} < .05 \cdot 2\abs{\mu}] \leq .003.
    \end{equation}
    
    On the other hand, applying Markov's inequality to \Cref{prop:verse} gives
    \begin{equation} 
        \Pr[\hav \wt{\btheta} > 1000 \mu^2/(1+s^2)] \leq .001,
    \end{equation}
    and we conclude that except with probability at most $.001$ we have
    \begin{equation}
        \hav \wt{\btheta} \leq 1000 \frac{\mu^2}{1+s^2} \leq 1000 \mu^2 ~~\implies~~ \abs{\sin(\btheta/2)} \leq \sqrt{1000} \abs{\mu} ~~\implies~~ \abs{\btheta/2} \leq \tfrac{\pi}{2}\sqrt{1000} \abs{\mu} < 50 \cdot 2\abs{\mu}.
    \end{equation}
    Hence
    \begin{equation}
        \Pr[\abs{\wt{\btheta}} > 50 \cdot 2\abs{\mu}] \leq .001.
    \end{equation}
    Finally, \Cref{fact:ourip} implies the fidelity bound $\abs{\braket{\bone|\bonef}}^2 \geq (1+\mu^2)/(1+s^2) \geq 1/(1+s^2)$, from which it is not hard to deduce 
    \begin{equation}    \label[ineq]{ineq:a2}
        \Pr[\abs{\btheta} > 50 \cdot 2\abs{\mu}] \leq .001 + s \leq .002
    \end{equation}
    under the assumption $s \leq .001$.
    (One can, e.g., use Helstrom's theorem~\cite{Hel76} for this.)
    Combining \Cref{ineq:a1,ineq:a2} completes the proof.
\end{proof}

As mentioned, by using the version of \Cref{thm:multi} with good constants (proved in \Cref{sec:details}), we can easily complete the proof of \Cref{thm:maintask} by appealing to phase estimation.  In fact, with this method we do not even need the ``$\abs{\mu} \leq 2\eps$'' part of ``case~(ii)'' in the theorem statement; just $\abs{\mu} \geq \eps$ is sufficient.
Thus we have the following slightly stronger form of \Cref{thm:maintask}:
\begin{theorem}[A stronger form of \Cref{thm:maintask}]  \label{thm:equivm}
    There is a computationally efficient quantum algorithm with the following properties:
    Given a parameter $\eps > 0$ and the code for a random variable~$\by$, promised to satisfy $s = \sqrt{\E[\by^2]} \leq \frac{1}{16}$,
    the algorithm uses $O(1/\eps)$ samples and distinguishes (with confidence at least~$2/3$) between the cases (i)~$\abs{\mu} \leq \eps/2$ and (ii)~$\abs{\mu} \geq \eps$, where $\mu = \E[\by]$.
\end{theorem}
\begin{remark}  \label{rem:constants0}
    The reader will notice that we have taken $s \leq \frac{1}{16}$ as a hypothesis here, whereas \Cref{thm:maintask} has $s \leq 1$. 
    However one may observe that the theorem's statement is insensitive to multiplying both~$\by$ and~$\eps$ by any fixed constant $0 < s_0 < 1$ (such as $s_0 = \frac{1}{16}$\footnote{For the sake of implementation it is nicer if the constant is a power of~$2$ so that adjusting the oracle~$\calY$ is simple.}); this only affects the sample complexity by a constant factor.
\end{remark}
\begin{proof}[Proof of \Cref{thm:equivm}]
    We perform Quantum Phase Estimation on the unitary~$\calU$ and starting state~$\ket{\bone}$, with accuracy parameter~$\eps/6$ and confidence parameter~$\delta = 1/9$, producing output~$\btheta'$.
    As described in \Cref{rem:QPE}, this can be done with $O(1/\eps)$ uses of controlled-$\calU$, which implies $O(1/\eps)$ uses of the code for~$\by$ (\Cref{rem:uses})
    The result is that, for $\btheta \sim \Theta_{\calU}(\ket{\bone})$ being the ``Idealized Phase Estimation'' output, there is a coupling such that $|\btheta - \btheta'| \leq \eps/6$ except with probability at most~$1/9$.
    Thus from \Cref{thm:multi}, we get that
    \begin{equation}
        \tfrac45 \abs{\mu} - \eps/12 \leq  |\btheta'/2|  \leq \tfrac54 \abs{\mu} + \eps/12
    \end{equation}
    except with probability at most~$\frac{1}{9} + \frac{2}{9} = 1/3$.
    Now on one hand, in case~(i) we have
    \begin{equation}
        \abs{\mu} \leq \eps/2  \implies |\btheta'/2| \leq (5/8)\eps + \eps/12  < .71\eps.
    \end{equation}
    On the other hand, in case~(ii) we have
    \begin{equation}
        \abs{\mu} \geq \eps  \implies |\btheta'/2| \geq (4/5) \eps - \eps/12 > .71 \eps.
    \end{equation}
    Thus we can distinguish the two cases with confidence at least~$2/3$ by deciding whether $\btheta' \gtrless 1.42 \eps$.
\end{proof}

\subsection{Sharper eigenvalue analysis} \label{sec:details}
In this section, we establish \Cref{thm:multi} with the explicit good constants.
The intuition behind the analysis is the following observation: Suppose for a moment that $\wt{\btheta}$~and~$\btheta$ from \Cref{prop:verse,prop:inverse} were identically distributed.  Then, writing $\bh = \frac{\sqrt{1+s^2}}{\abs{\mu}} \sqrt{\hav \btheta}$, these would imply 
\begin{equation}
    \E[(\bh - \bh^{-1})^2] = \frac{1+s^2}{\mu^2}\E[\bh] + \frac{\mu^2}{1+s^2}\E[\bh^{-1}] - 2\E[1] = 1 + 1 - 2 = 0,
\end{equation}  
and hence
\begin{equation}
    \bh - \bh^{-1} \equiv 0 \quad\implies\quad \bh \equiv 1 \quad\implies\quad \sqrt{\hav \btheta} = \abs{\sin(\btheta/2)} \equiv \frac{\abs{\mu}}{\sqrt{1+s^2}}.
\end{equation}
The next theorem (of which \Cref{thm:multi} is a corollary) makes this idea rigorous by taking care of the fact that we don't quite have $\wt{\btheta} \not \equiv \btheta$:
\begin{theorem} \label{thm:a1}
    Fix any $C \geq 1$ and assume  $s \leq \frac{1}{C}$.
    Then for $\btheta \sim \Theta_{\RR}(\ket{\bone})$,
    \begin{equation}
        \Pr\bracks*{\abs{\sin(\btheta/2)} \not \in \frac{\abs{\mu}}{\sqrt{1+s^2}} \cdot \bracks*{\frac{1}{1+Cs}, \frac{1}{1-Cs}}} \leq \frac{2}{C^2}.
    \end{equation}
\end{theorem}
The proof will use a numerical lemma:
\begin{lemma}   \label{lem:weird}
    For real numbers $h, \lambda$ with $h  > 0$, it holds that
    $\displaystyle
            (1-h^{-1})^2 \leq (1-\lambda)^2 + (\lambda h - h^{-1})^2.
    $
\end{lemma}
\begin{proof}
    The difference of the two sides is $\lambda^2(1-h)^2 + 2h(\lambda - h^{-1})^2  \geq 0$.
\end{proof}
\begin{proof}[Proof of \Cref{thm:a1}.]
    As in our proof of \Cref{thm:multi} with worse constants, we may assume without loss of generality that $\mu, \btheta \neq 0$ always holds.
    We use the notation of \Cref{not:eigen} and \Cref{prop:bhatt}, letting 
    $q_1, \dots, q_D$ (respectively, $r_1, \dots, r_D$) denote the probabilities of $J_{\RR}(\ket{\bone})$ (respectively, $J_{\RR}(\ket{\bonef})$) and $\lambda_j = \sqrt{r_j/q_j}$. 
    Now combining \Cref{fact:ourip} with \Cref{ineq:hqq} from \Cref{prop:bhatt} gives
    \begin{equation}    \label[ineq]{ineq:hqq2}
        \sum_j q_j (1 - \lambda_j)^2
        = 2 \parens*{1 - \sum_j q_j \lambda_j} 
        \leq 2(1 - \abs{\braket{\bone|\bonef}}) 
        = 2(1 - \sqrt{1+\mu^2}/\sqrt{1+s^2})
        \leq 2(1 - 1/\sqrt{1+s^2}) \leq s^2.
    \end{equation}
    At the same time, if we define
    \begin{equation}
        h_j = \frac{\sqrt{1+s^2}}{\abs{\mu}} \sqrt{\hav \theta_j},
    \end{equation}
    we can restate \Cref{prop:inverse} (which has its technical assumption satisfied) and \Cref{prop:verse}   as 
    \begin{equation}    \label{eqn:yo}
        \sum_j q_j h_j^{-2} = 1, \quad \text{and} \quad \sum_j r_j h_j^2 = 1  \iff \sum_j q_j \lambda_j^2 h_j^2 = 1.
    \end{equation}
    Let $\bj  \sim J_{\calU}(\ket{\bone})$, so $\Pr[\bj = j] = q_j$.
    We may express $\btheta = \theta_{\bj}$, and also write $\bh = h_{\bj}$ and $\blambda = \lambda_{\bj}$.
    Then we may summarize \Cref{eqn:yo} and \Cref{ineq:hqq2}  as
    \begin{equation}
        \E[\bh^{-2}] = \E[\blambda^2 \bh^2] = 1, \qquad \E[(1-\blambda)^2] = 2(1-\E[\blambda]) \leq s^2
    \end{equation}
    (where we used $\E[\blambda^2] = 1$). 
    These imply
    \begin{equation}
        \E[(\blambda \bh - \bh^{-1})^2] = \E[\blambda^2 h^2] + \E[\bh^{-2}] - 2\E[\blambda] = 2(1 - \E[\blambda]) \leq s^2.
    \end{equation}
    Now applying \Cref{lem:weird} in expectation and using the above facts, we get
    \begin{equation}
        \E[(1-\bh^{-1})^2] \leq \E[(1-\blambda)^2] + \E[(\blambda \bh - \bh^{-1})^2] \leq s^2 + s^2 = 2s^2.
    \end{equation}
    Thus by Markov's inequality, except with probability at most~$2/C^2$, we have
    \begin{equation}
        (1 - \bh^{-1})^2 \leq (Cs)^2  \implies \abs{1-\bh^{-1}} \leq Cs \implies \bh \in  \bracks*{\frac{1}{1+Cs}, \frac{1}{1-Cs}},
    \end{equation}
    (recall $Cs \leq 1$).
    Putting in the definition of $\bh = \frac{\sqrt{1+s^2}}{|\mu|} \sqrt{\hav \btheta}$ and recalling $\hav \theta = \sin^2(\theta/2)$ completes the proof.
\end{proof}

Note that given the $1 \pm O(s)$ error range of the preceding theorem, the distinctions between $\abs{\sin(\btheta/2)}$ and $\abs{\btheta}^2$ and between $\mu^2/(1+s^2)$ and $\mu^2$ are more minor.
Thus the preceding theorem essentially gives that $|\btheta/2| = |\mu| \pm O(|\mu| s)$ with high probability.
With some slightly tedious numeric estimates, we can get the following more usable corollary; it immediately implies \Cref{thm:multi} with its strong constants by taking $t = 1$ (and using $1 - \frac{3}{16} \geq \frac45$):
\begin{corollary}   \label{cor:done}
    Fix any $t \geq 1$.  
    Then \Cref{thm:multi} holds with $s_0 = \frac{1}{16t}$, $c_0 = 1 - 3ts_0$, $C_0 = 1+4ts_0$, and $\delta_0 = \frac{2}{9t^2}$. 
\end{corollary}
\begin{proof}
Given $t$, write $\eta = ts$ and note that $\eta \leq \frac{1}{16}$ assuming $s \leq s_0$.  
    Now selecting $C = 3t$ in \Cref{thm:a1},
    we get that  except with probability at most~$\frac{2}{9t^2}$:
    \begin{align} 
        2 \cdot \abs{\sin({\btheta}/2)} &\geq 2\abs{\mu} \cdot \frac{1}{\sqrt{1+(\eta/t)^2}}\cdot \frac{1}{1+3\eta},  &
        2 \cdot \abs{\sin({\btheta}/2)} &\leq 2\abs{\mu} \cdot \frac{1}{\sqrt{1+(\eta/t)^2}}\cdot \frac{1}{1-3\eta} \\
        &\geq 2\abs{\mu} \cdot \frac{1}{\sqrt{1+\eta^2}}\cdot \frac{1}{1+3\eta} \geq 2\abs{\mu} \cdot (1-3\eta) & & \leq 2\abs{\mu} \cdot \frac{1}{1-3\eta} \label[ineq]{ineq:yuck}
    \end{align}
    Since $\abs{\btheta} \geq 2 \cdot \abs{\sin(\btheta/2)}$, we now have the needed lower bound for \Cref{ineq:new},
    \begin{equation}
        \Pr[\abs{\btheta} < (1-3ts) \cdot 2\abs{\mu}] \leq \frac{2}{9t^2}.
    \end{equation}
    
    As for the upper bound, let us first weakly observe that $\abs{\mu} \leq s \leq s_0 \leq \frac{1}{16}$, which together with $\eta \leq \frac{1}{16}$ means \Cref{ineq:yuck} implies $2 \cdot \abs{\sin(\btheta/2)} \leq 2\cdot \frac{1}{16} \cdot \frac{1}{1-3(1/16)} = \frac{2}{13}$.  
    On this range of~$\abs{\btheta}$, it holds that
    $\abs{\btheta} \leq (13 \sin^{-1} \tfrac{1}{13}) \cdot 2 \cdot \abs{\sin(\btheta/2)}$.
    Combined with \Cref{ineq:yuck} and $\eta \leq \frac{1}{16}$ this yields
    \begin{equation}    \label[ineq]{ineq:theta3}
        \abs{\btheta} \leq 32 \sin^{-1}(\tfrac{1}{13})  \cdot \abs{\mu}  \quad \implies \quad \tfrac{1}{24}\abs{\btheta}^3 \leq .63 \abs{\mu}^3 \leq .63 \cdot \frac{1}{16} \abs{\mu} ts \leq .04 \abs{\mu} \eta,
    \end{equation}
    where we used $\abs{\mu} \leq s \leq s_0 \leq \frac{1}{16} \leq \frac{1}{16}t$.
    Finally we come to the main use of \Cref{ineq:yuck}:
    \begin{equation}
        \abs{\btheta} - \tfrac{1}{24}\abs{\btheta}^3 \leq 2 \cdot \abs{\sin(\btheta/2)} \leq 2\abs{\mu} \cdot \frac{1}{1-3\eta} \leq 2\abs{\mu} \cdot (1 + \tfrac{48}{13} \eta)
    \end{equation}
    using $\eta \leq \frac{1}{16}$ again.
    From this, \Cref{ineq:theta3}, and $\frac{48}{13} + .02 \leq 4$, we deduce 
    \begin{equation}
        \Pr[\abs{\btheta} > (1+4ts) \cdot 2\abs{\mu}] \leq \frac{2}{9t^2},
    \end{equation}
    the needed upper bound in \Cref{ineq:new}.
\end{proof}

\subsection{A more elementary algorithm}    \label{sec:elementary}
One could argue that our algorithm's use of Quantum Phase Estimation is overkill: we have high-probability bounds for the location of~$\btheta$ (so doing estimations in superposition is not needed), and this location is restricted to two narrow, separated regions: $\btheta \approx 2\abs{\mu}$ or $\btheta \approx 0$. 
Thus it is possible to use a more ``elementary'', Grover-like method to go from \Cref{thm:multi} to \Cref{thm:maintask}, as we now demonstrate.

We first show this by appealing to the very strong constants achieved in \Cref{cor:done}; we then sketch how even the simply obtained constants from \Cref{eqn:simple} suffice.

So suppose first we take $t = 10$ in \Cref{cor:done}, leading to $\delta_0 = 2/900 \leq .003$.  
We can also take our upper bound~$s_0$ on~$s$ as small as we please (see \Cref{rem:constants0}); let us therefore take it small enough that $4t s_0 \leq .01$.
We thereby obtain from \Cref{cor:done} that
\begin{equation}    \label[ineq]{ineq:calcs}
    \Pr_{\btheta \sim \Theta_{\calU}(\ket{\bone})}[.99 \cdot 2\abs{\mu} \leq \abs{\btheta} \leq 1.01 \cdot 2 \abs{\mu}] \geq .997.
\end{equation}
Recall we are trying to distinguish the cases (i)~$\abs{\mu} \leq \eps/2$ and (ii)~$\eps \leq \abs{\mu} \leq 2\eps$.
Note also that $\abs{\mu} \leq s \leq s_0 < .0003$ by assumption, meaning we can assume~$\eps \leq .0003$ without loss of generality. 
Suppose we now take 
\begin{equation}
    T = \lfloor \pi/(3\eps) \rfloor    
\end{equation}
(noting that $\eps \leq .0003$ means the floor changes~$T$'s value by a factor of at most~$1.0003$).
Then \Cref{ineq:calcs} implies
\begin{equation}
    \Pr_{\btau \sim \Theta_{\calU^T}(\ket{\bone})}[.989 \cdot \pi/(3\eps) \cdot 2\abs{\mu} \leq \abs{\btau} \leq 1.01 \cdot \pi/(3\eps) \cdot 2\abs{\mu}] \geq .997.
\end{equation}
(We changed $.99$ to $.989$ to account for the floor on~$T$.)
So except with probability at most $.003$ we have the following:
\begin{align}
    \text{case (i)} &\quad \implies\quad \abs{\btau} \leq 1.01 \cdot \pi/3 &\implies\quad \cos \btau &\geq +.49, \\
    \text{case (ii)} &\quad \implies\quad .989 \cdot 2\pi/3 \leq \abs{\btau} \leq 1.01 \cdot 4\pi/3 &\implies\quad \cos \btau &\leq -.46.
\end{align}
Now suppose we perform the Hadamard Test on~$\ket{\bone}$ with the unitary~$\calU^T$, whose application uses the code for~$\by$ only $O(T) = O(1/\eps)$ times.  (Recall this means adjoining~$\ket{+}$ to $\ket{\bone}$, applying \mbox{controlled-$\calU^T$}, and then measuring the new qubit in the $\ket{\pm}$ basis.)
With probability at least~$.997$, we get back a random variable~$\bb \in \{\pm 1\}$ with expectation~$\cos \btau$.  
Thus in case~(i) we get $\bb = +1$ with probability at least $\frac12 + \frac12(+.49) - .003 \geq 2/3$, and in case~(ii) we get $\bb = +1$ with probability at most $\frac12 + \frac12(-.46) + .003 \leq 1/3$.
Thus we can distinguish the two cases with confidence~$2/3$ based on the measurement outcome~$\bb$ of the Hadamard Test.
An example depiction of the process underlying this more elementary algorithm is shown in \Cref{fig:aa,fig:bb}.\\

We now sketch how one can go from \Cref{thm:multi} to \Cref{thm:maintask} in a similarly elementary way even with the worse constants from \Cref{eqn:simple}. 
For this, we will need to weaken the statement of \Cref{thm:maintask} so that case~(i) is ``$\abs{\mu} \leq c_1 \eps$'' for a very small constant~$c_1 > 0$.  This is not without loss of generality, as changing the upper bound on~$s$ is.  Nevertheless,  as we will show (see \Cref{rem:constants}), this does not affect our ability to deduce \Cref{thm:main0} from \Cref{thm:maintask}.

Recall that with the constants from \Cref{eqn:simple} we have
\begin{equation}
    \Pr_{\btheta \sim \Theta_{\calU}(\ket{\bone})}[.05 \cdot 2\abs{\mu} \leq \abs{\btheta} \leq 50 \cdot 2\abs{\mu}] \geq .995
\end{equation}
provided $s \leq .001$.
Now suppose we are in case~(ii), $\eps \leq \abs{\mu} \leq 2\eps$, so that $.1\eps \leq \abs{\btheta} \leq 200\eps$ except with probability~$.005$.  The trick is to divide this range of multiplicative-width~$2000$ into, say, $\lceil \log_2 2000 \rceil = 11$ intervals of multiplicative-width at most~$2$:  say, $[.1 \eps, .2\eps], [.2 \eps, .4 \eps]$, \dots $[100\eps, 200\eps]$.
The most frequently encountered such interval for~$\btheta$ occurs with probability at least $.995/11 \geq .09$.  Now suppose we take $11$ different values of~$T = \Theta(1/\eps)$ so that the scaled-by-$T$ intervals approximate~$[2\pi/3, 4\pi/3]$.  (Note that each~$T$ is at most $22/\eps$.)  Then for at least one such~$T$, the Hadamard Test applied to $\ket{\bone}$ and $\calU^T$ will output~$\bb = +1$ with probability slightly bounded away from~$1$; at most $.09 \cdot q + .91$ for $q \approx \frac12 + \frac12 \max\{\cos(2\pi/3), \cos(4\pi/3)\} = .25$, in particular, at most~$.94$.

On the other hand, suppose we take case~(i) in \Cref{thm:maintask} to be ``$\abs{\mu} \leq .0001\eps$''.
Then one can infer that except with probability at most~$.995$ over the outcome of $\btheta \sim \Theta_{\calU}(\ket{\bone})$ we will have $\abs{\btheta} \leq .01 \eps$ and hence $\abs{\btau} \leq .01T \leq .22$ for \emph{each} of the~$11$ choices of~$T$.  Thus $\cos \btau \geq \frac12 + \frac12 \cos(.22) \geq .987$ and we conclude the Hadamard Test will report~$\bb = +1$ with probability at least $.987 - .005 \geq .98$.

In summary, we have $11$ different ``coins''~$\bb$, with the guarantee that in case~(i) all come up Heads ($+1$) with probability at least~$.98$, and in case~(ii) at least one comes up Heads with probability at most~$.94$.  The two cases may therefore be distinguished with confidence at least~$2/3$ using a constant number of ``coin flips'' (each of which uses $O(T) = O(1/\eps)$ samples).

\begin{remark}
    We observe that the fundamental feature of the constants from \Cref{eqn:simple} that make them acceptable is that $\delta_0 \ll 1/\log(C_0/c_0)$.  Indeed, our simpler proof of \Cref{thm:multi} achieves $\delta_0 \leq O(c_0 + 1/C_0)$.
\end{remark}

\begin{figure}[H] \centering
    \includegraphics[width=.32\textwidth]{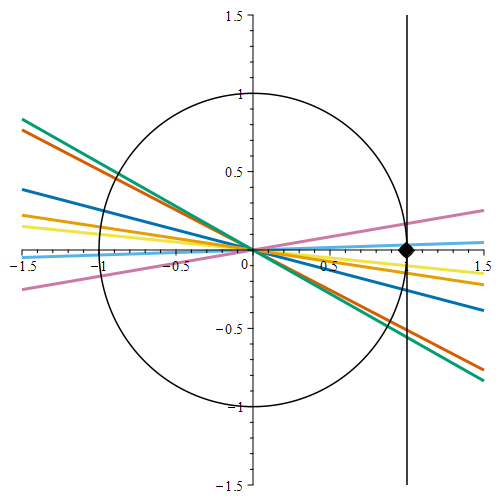}\\
    \includegraphics[width=.32\textwidth]{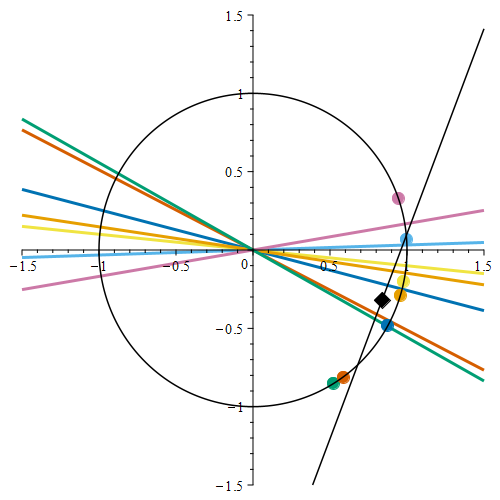}
    \qquad
    \includegraphics[width=.32\textwidth]{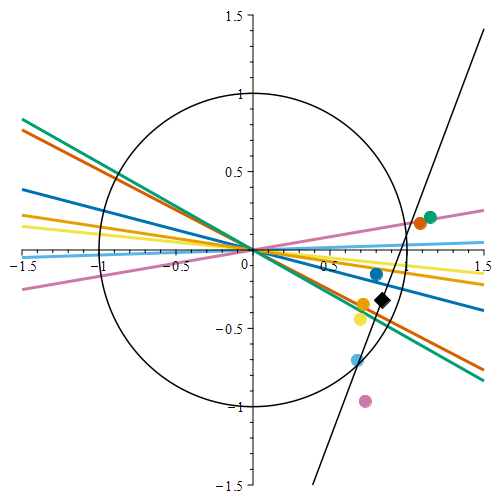}\\
    \includegraphics[width=.32\textwidth]{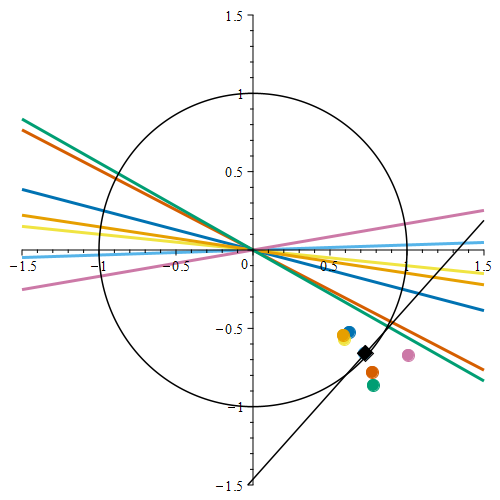}
    \qquad
    \includegraphics[width=.32\textwidth]{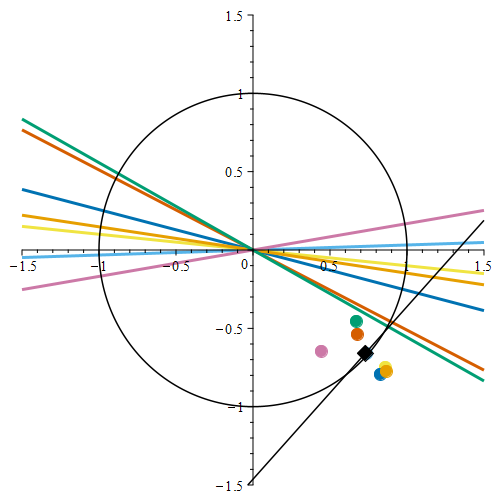}
        \caption{
        An illustration of $\calU$'s action. 
        We have the uniform distribution~$p$ over $D = 7$ outcomes,  $y_1, \dots, y_7$, with values approximately $-.169$, $-.032$, $.101$, $.148$, $.258$, $.511$, $.557$ (pink, light blue, light yellow, orange, dark blue, dark orange, green).
        We have $\mu = \E_p[\by] \approx .196 \approx \pi/(2 \cdot 8)$, and $s^2 = \E_p[\by^2] = .1$. 
        The colored line corresponding to~$y_\ell$ is at angle~$\alpha_\ell$; i.e., it passes through $1 - \i y_\ell$ in the complex plane.
        The five diagrams above show $\ket{\bone}$, $\RotF \ket{\bone}$, $\Refl \cdot \RotF \ket{\bone} = \calU \ket{\bone}$, $\RotF \cdot \calU \ket{\bone}$, and $\Refl \cdot \RotF \cdot \calU \ket{\bone} = \calU^2 \ket{\bone}$ (from top to bottom, left to right).
        The black diamond shows the barycenter of the points (through which the reflection occurs), and the black line is orthogonal to it.
        \label{fig:aa}
    }
\end{figure}
\begin{figure}[H] \centering
    \includegraphics[width=.32\textwidth]{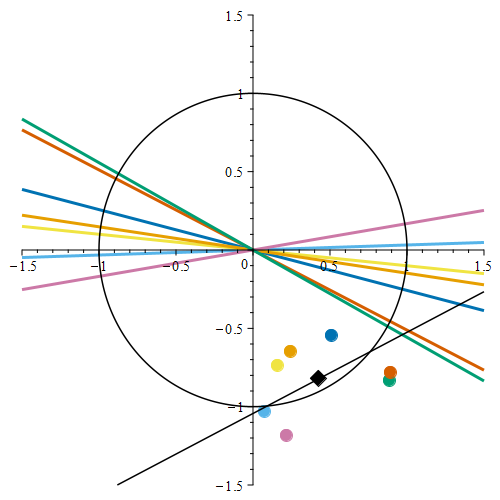}
    \includegraphics[width=.32\textwidth]{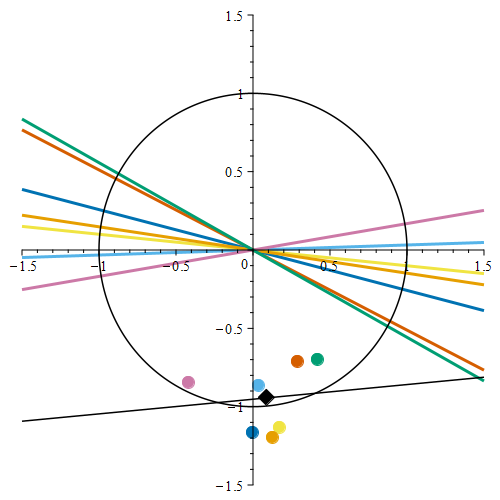}
    \includegraphics[width=.32\textwidth]{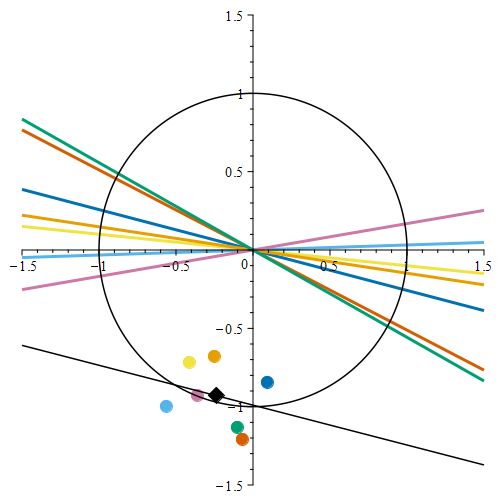}\\
    \includegraphics[width=.32\textwidth]{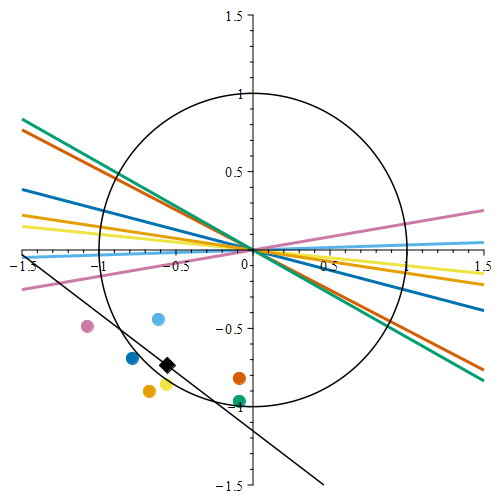}
    \includegraphics[width=.32\textwidth]{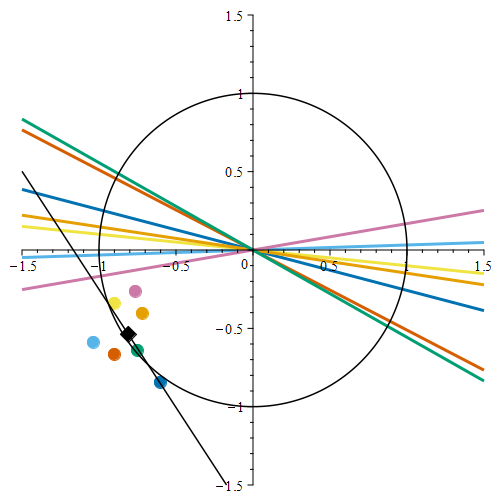}
    \includegraphics[width=.32\textwidth]{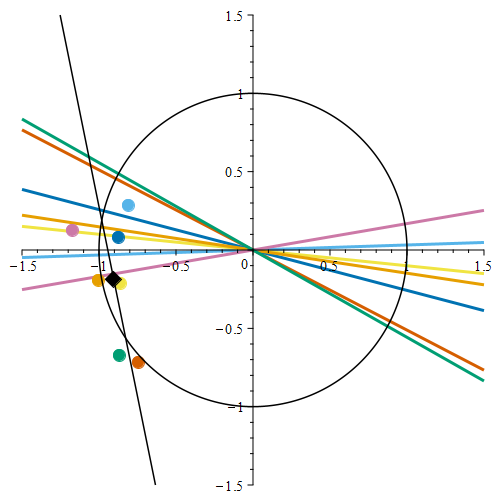}
    \caption{Continuing \Cref{fig:aa}, we illustrate $\calU^t \ket{\bone}$ for $t = 3 \dots 8$.
    The result of applying the Hadamard Test at any time is a $\{\pm 1\}$-valued random variable with mean equal to the horizontal displacement of the black diamond.
    After~$8 \approx  \pi/(2\abs{\mu})$ applications of~$\calU$,  this displacement is close to~$-1$.
        \label{fig:bb}
    }
\end{figure}

\subsection{Describing all eigenvalues and eigenvectors} \label{sec:rankone}
Here we give a geometric description of all the eigenvalues and eigenvectors of~$\calU$.
Their description is simple enough that one might discover them through intuition.
An alternative route (to the eigenvalues, at least) is to observe that 
\begin{equation}
    \calU =  \Refl \cdot \RotF = \text{$\U$}(2\text{$\ketzero\!\brazero$} - \text{$\Id$})\text{$\U^\dagger$} \cdot \RotF = 2\text{$\U$}\text{$\ketzero\!\brazero$} \cdot \RotF - \RotF
\end{equation}
is a rank-one update of the diagonal matrix~$\RotF$ (up to a minus sign). 
As such, one can give an explicit rational expression (see e.g.~\cite{Ion01}) whose roots are the eigenvalues of~$\calU$. 
By working through the details one can arrive at the below geometric description of the eigenvalues.

It will actually be slightly more convenient to consider $\calU^\dagger = \RotF^\dagger \cdot \Refl$, which has the same eigenvectors as~$\calU$ and the complex-conjugated eigenvalues.
To seek the eigenvectors of $\calU^\dagger$, consider the random variable~$\ket{\bone + \bii \by}$, which we saw is an eigenvector of~$\calU$ (and~$\calU^\dagger$) of eigenvalue~$1$ if $\mu = \E_p[\by] = 0$.
Plot each value of~$\ket{\bone + \bii \by}$ (i.e., $1+y_\ell$) as a point in the complex plane, along with the line passing through it and the origin, as in the diagram on the left in \Cref{fig:eigs}. 
Now imagine slowly rotating all the lines, always marking the points where they cross the vertical line corresponding to real-part~$1$.
Also, keep track of the mean of these points (that is, the weighted mean under~$p$), which will naturally also be on the same vertical line. 
Pause rotation whenever this barycenter touches the real axis (i.e., becomes~$1 + 0\i$), as in the diagram on the right in \Cref{fig:eigs}.
Say that after pausing we have rotated by~$\theta$ and the current points form the random variable~$\ket{\bone + \bii \by'}$. 
Then we claim that $\ket{\bone + \bii \by'}$ is an eigenvector of~$\calU^\dagger$ with eigenvalue $e^{-\i \cdot 2\theta}$ (hence an eigenvector of~$\calU$ with eigenvalue $e^{\i \cdot 2\theta}$).

\begin{figure}[H]\centering
    \includegraphics[width=.48\textwidth]{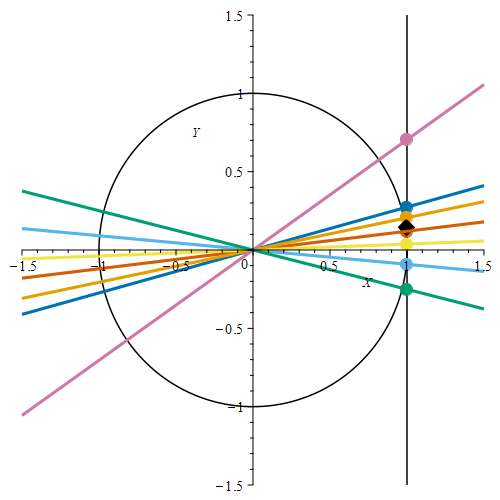}\quad
    \includegraphics[width=.48\textwidth]{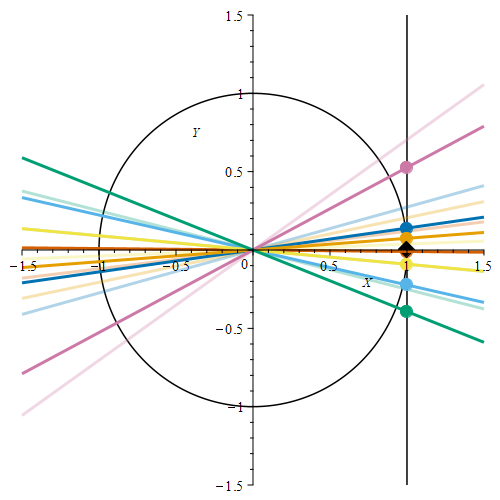}
    \caption{An example with the uniform distribution~$p$ on $D = 7$ outcomes again. 
    The $y_\ell$ values are $.560$, $.258$, $.057$, $-.045$, $-.088$, $-.250$, $-.494$ 
    (yellow, light blue, dark blue, green, pink, light orange, dark orange).
    In contrast to \Cref{fig:aa,fig:bb}, here the associated colored lines pass through $1 + \i y_\ell$ (rather than $1 - \i y_\ell$).
    As  before, the black diamond depicts the barycenter.
    On the left, we have the initial points $\ket{\bone + \bii \by}$.
    On the right, after rotating the lines slightly (with the original lines shown with light color), the new points' barycenter is on the real axis.
    Twice the angle of this rotation is an eigenphase of~$\calU$, with the point locations on the right forming the associated eigenvector.
    \label{fig:eigs} }
\end{figure}

To verify this claim, first observe that $\Refl \ket{\bone + \bii \by'} = \ket{\bone - \bii \by'}$, since $\E_p[\by'] = 0$. 
Next, recall that $\RotF^\dagger$ rotates the $\ell$th point by $2\arctan y_\ell$, so that $1 - \i y_\ell$ moves to $1 + \i y_\ell$.
Since $1 - \i y'_\ell$ is at angle $-\theta$ from~$1-\i y_\ell$, it follows that $\RotF^\dagger$ moves the $\ell$th point so that it is at angle~$-\theta$ from $1+\i y_\ell$, and hence angle $-2\theta$ from its starting location of $1 + \i \by'$.  
Thus we see that the composition $\calU^\dagger = \RotF^\dagger \cdot\Refl$ indeed multiplies $\ket{\bone + \bii \by'}$ by $e^{-\i 2\theta}$, as claimed.

So far we have explained how to find one eigenvector/value of~$\calU^\dagger$.  
To find more, we simply keep rotating the lines, waiting for ``black diamond'' to cross the real axis; we show the next two such occurrences in \Cref{fig:eigs2}.  
Note that as we rotate (clockwise, in the figures), the barycenter moves monotonically downward until such time as one of the colored lines rotates to a vertical position (taking the associated colored point's height to~$-\infty$); at this point, the black diamond's height ``wraps around'' to~$+\infty$, and then continues monotonically downward.
From this observation, it is easy to see that if $\by$'s $D$ values are all distinct, we will get the full complement of~$D$ distinct eigenvectors.  
(Otherwise eigenvalues will occur with multiplicity, but one can reduce to the distinct case by infinitesimal perturbations.)

\begin{figure}[H] \centering
    \includegraphics[width=.48\textwidth]{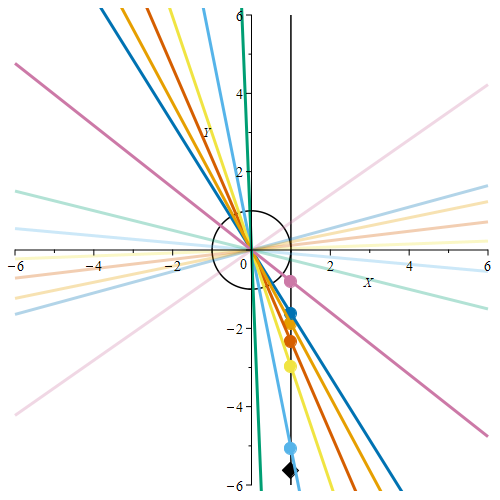}\quad
    \includegraphics[width=.48\textwidth]{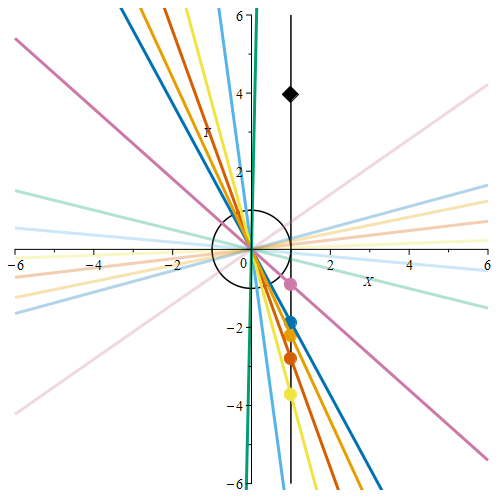} \\
    \includegraphics[width=.48\textwidth]{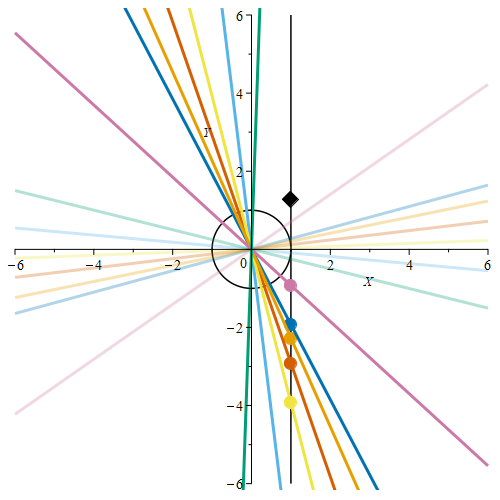}\quad
    \includegraphics[width=.48\textwidth]{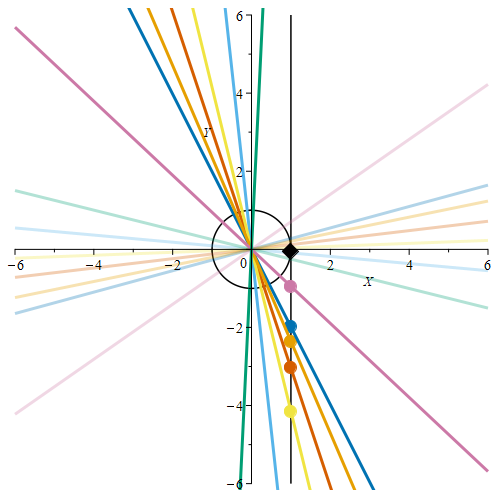} \\
    \caption{\label{fig:eigs2} Continuing from \Cref{fig:eigs}, we depict further rotations (with the diagram zoomed out by a factor of~$4$).  As the green line approaches vertical, the height of the barycenter (black diamond) approaches $-\infty$. As the green line crosses vertical, the black diamond's height becomes~$+\infty$ and begins descending again, until we reach a second eigenvector in the final diagram.
    The light blue point is outside the picture in the final three diagrams, and the green point is outside the picture in all four diagrams.}
\end{figure}

We observe that each of the eigenvectors $\ket{\bone + \bii \by'}$  we have described is unnormalized; the unit-length version of it is $\ket{\bone + \bii \by'}/\sqrt{1 + \E_p[(\by')^2]}$.  
Hence the overlap of our algorithm's starting vector, $\ket{\bone}$, with the eigenvalue is $1/\sqrt{1 + \E_p[(\by')^2]}$.
In the figures above, we have depictions of unnormalized eigenvectors as colored points on the $\Re = 1$ line with black diamond on the real axis.
Thus the overlap of the \emph{normalized} eigenvector with $\ket{\bone}$ is high if and only if the colored points' vertical heights are not ``too extreme'' on average.
In fact, using this viewpoint, it's not too hard to design initial random variables~$\by$ for which \emph{no single}  eigenvector of eigenphase $|\theta| \approx 2\mu$ has large overlap with~$\ket{\bone}$; only the whole subspace of them does.
This justifies why our analysis in \Cref{sec:eigo} cannot work (as in Grover's algorithm) by simply identifying one or two eigenvectors of eigenphase around $2|\mu|$ with which $\ket{\bone}$ has large overlap.

\section{Consequences}\label{sec:consequences}
\newcommand{\loweps}{c}

Having established \Cref{thm:maintask} which achieves the ``Main Task'', we will now layer on top a sequence of improvements that will culminate in our  
solution to the Mean Estimation problem from \Cref{thm:main0}. 
All but the last of these improvements is simply a classical reduction based on ``standard'' ideas (binary search, successive halving, etc.).
By the end of these classical improvements we will have achieved the results in \Cref{table:prior} up through~\cite{Mon15}, but with the improved (optimal) sample complexity of~$O(n)$.
(We remark that several of the aforementioned ``standard'' ideas appeared earlier in the works from \Cref{table:prior}.)
The final step, which gets us to the optimal Mean Estimation algorithm, requires combining several of the preceding reductions with the quantum Quantile Finding algorithm of Hamoudi~\cite{Ham21}.
As discussed in \Cref{rem:gatecxty} and \Cref{sec:gates}, all of our algorithms are also gate efficient; however for clarity of exposition, we will focus only on query (sample) complexity in this section.\\

We describe the improvements below as a sequence of problems to be solved.
In all of these problems, the setup remains the same: We have access to ``the code'' for a random variable~$\by$, 
we write
\begin{equation}
    \mu = {\E}_p[\by], \qquad \sigma^2 = {\Var}_p[\by], \qquad s^2 = {\E}_p[\by^2],
\end{equation}
and we are trying to solve the given problem with success probability at least~$2/3$.

\paragraph{The confidence parameter, and the ``log\,log trick''.} 
All of the following problems have the usual feature that the arbitrarily selected confidence parameter of~$2/3$ can be boosted to $1-\delta$ at the expense of $O(\log(1/\delta))$ repetitions (followed by taking the majority/median answer).
This allows us to chain together constantly many solutions at constant expense; we will omit explicit mention of this standard technique. However, in two cases we will need the following ``log\,log trick'' (which has certainly been used before, but doesn't seem to have a standard name).
 
Assume we plan to solve a sequence of problems with decreasing ``accuracy'' parameters $1 \geq \eta_1 \geq \eta_2 \geq \cdots \geq \eta_{t} \geq \eta^*$.
Here the values of~$\eta_1$ and $\eta^*$ should be fixed in advance, but we do not require that the other~$\eta_j$'s are; it is acceptable if~$\eta_{j+1}$'s value is chosen only after the solution for accuracy~$\eta_j$ is found. 
However we \emph{do} always require that $\eta_{j+1} \leq \eta_j/R$ for some fixed constant~$R < 1$.
It is also assumed that that solving a problem with accuracy~$\eta_j$ and confidence $1-\delta_j$ can be done at a ``cost'' of  $O(1/\eta_j) \cdot \log(1/\delta_j)$.

If we could ignore the issue of confidence, the costs would be upper-bounded by~$O(\cdot)$ of a geometric series of ratio~$R > 1$, beginning at~$1$ and ending just past~$1/\eta^*$.
Hence the total cost would be bounded by the \emph{final} cost of $O(1/\eta^*)$, up to a constant factor depending only on~$R-1$.  
Our goal is to achieve this cost, while properly taking into account the confidence parameter.
If, naively, we decided to take $\delta_j = \delta$ for all~$j$, then we would have to take $\delta \leq 1/(3T)$, where $T = O(\log(1/\eta^*))$ is an upper bound on the number of problems solved.  This would incur an extra multiplicative cost of~$O(\log(1/\delta)) = O(\log T) = O(\log \log (1/\eta^*))$.

To evade this extra ``log\,log'' factor, we may solve the~$j$th problem with confidence parameter, say,
\begin{equation}
    \delta_j = \exp(-C (\eta_j/\eta^*)^{1/2}), 
\end{equation}
where $C = C(R)$ is a certain constant.
(Here the exponent~$1/2 \in (0,1)$ was chosen arbitrarily.)
Note that the algorithm only needs to know~$\eta_j, \eta^*$ to set this~$\delta_j$, not the values of $\eta_{j+1}, \eta_{j+2}, \dots$.
Now on one hand, if the number of stages ends up being~$t$, the union bound implies the total failure probability is at most
\begin{equation}
    \sum_{j=1}^t \delta_j 
    = \sum_{j=1}^t \exp(-C (\eta_j/\eta^*)^{1/2})
    \leq \exp(-C) + \exp(-C R^{1/2}) + \exp(-C R^{2/2}) + \exp(-C R^{3/2}) + \cdots \leq 1/3,
\end{equation}
provided~$C = C(R)$ is large enough.  
On the other hand, the total cost is~$O(\cdot)$ of
\begin{equation}
    \sum_{j=1}^t (1/\eta_j) \cdot \tfrac12 C (\eta_j/\eta_t)^{1/2}
    = O((1/\eta_t)^{1/2}) \cdot \sum_{j=1}^t (1/\eta_j)^{1/2}
    \leq O((1/\eta_t)^{1/2}) \cdot O((1/\eta_t)^{1/2}) = O(1/\eta_t),
\end{equation}
where the inequality used that the sum is bounded by a geometric series (of ratio $R^{1/2} > 1$) with final term~$1/\eta_t$.

\subsection{The classical reductions}
We begin with the Main Task:
\begin{prob}[Main Task]  \label{prob:problem1}
    \begin{tabbing}
        Input: \quad\= Parameter $0 < \eps \leq 1$.\\
        Promise: \> $\second \leq 1$; and, either (i)~$\abs{\mu} \leq \loweps\eps$ or else (ii)~$\eps \leq \abs{\mu} \leq 2\eps$ holds.\\
        Output: \>  Which of (i) or~(ii) holds.
      \end{tabbing}
\end{prob}
Our \Cref{thm:maintask} shows that this problem, with $c = 1/2$, can be solved with $O(1/\eps)$ uses of the code for~$\by$.
\begin{remark} \label{rem:constants}
    Here we wrote ``$c$'' more generally so we can illustrate that any universal constant $0 < c < 1$ will be acceptable.
\end{remark}

Next we show that a solution to the above problem can be used to solve a slightly more general problem where we have to decide if $\mu$ is close to some general target $\widehat\mu$, not necessarily~$0$:

\begin{prob}[Main Task given target $\widehat{\mu}$]\label{prob:problem2}
    \begin{tabbing}
        Input: \quad\= Parameter $0 < \eps \leq 1$, and preliminary estimate~$\wh{\mu} \in [-1,1]$.\\
        Promise: \> $\second \leq 1$; and, either (i)~$\abs{\widehat{\mu} - \mu} \leq \loweps\eps$ or else (ii)~$\eps \leq \abs{\widehat{\mu} - \mu} \leq 2\eps$ holds.\\
        Output: \>  Which of (i) or~(ii) holds.
      \end{tabbing}
\end{prob}

\begin{lemma}  
    We can solve \Cref{prob:problem2} with $O(1/\eps)$ queries  to the code for~$\by$.
\end{lemma}
\begin{proof}
    To solve \Cref{prob:problem2} for general~$\wh{\mu} \in [-1,1]$, let $\by' = \by - \wh{\mu}$.\footnote{Given the code for~$\by$, it is easy to produce the new code for~$\by'$. We defer all similar such observations to \Cref{sec:gates} on gate complexity.}  
    This has
    \begin{equation}
        (\second')^2 \coloneqq {\E}_p[(\by')^2] \leq 2{\E}_p[\by^2] + 2\wh{\mu}^2 \leq 2\cdot 1 + 2\cdot 1^2 = 4,
    \end{equation}
    where we used $(a-b)^2 \leq 2a^2 + 2b^2$.
    So if we further define $\by'' = \by'/2$, we get $(\second'')^2 \coloneqq \E_p[(\by'')^2] \leq 1$.
    Now it suffices to apply our solution for \Cref{prob:problem1} to~$\by''$, with~$\eps/2$ in place of its~$\eps$.
\end{proof}

The next upgrade is to actually estimate the mean of~$\by$, using use our solution to the general decision problem above.  
The idea is to use a form of binary search; this will in addition need the log\,log trick.

\begin{prob}[Mean estimation, promised second moment at most 1] \label{prob:problem3}
    \begin{tabbing}
      Input: \quad\= Parameter $0 < \eps \leq 1$.\\
      Promise: \> $\second  \leq 1$.\\
      Output: \> An estimate $\widehat{\mu}$ such that $\abs{\widehat\mu - \mu} \leq \eps$.
    \end{tabbing}
\end{prob}

\begin{lemma}   \label{lem:us}
    We can solve \Cref{prob:problem3} with $O(1/\eps)$ queries to the code for~$\by$.
\end{lemma}
\begin{proof}
    Given our algorithm for \Cref{prob:problem2}, repeating it $O(\log (1/\delta))$ times and taking the majority answer yields the ability to do the following:
    \begin{multline}
        \text{For any~$\wh{\mu}$, assuming $\abs{\widehat{\mu} - \mu} \leq c\eps'$ or $\eps' \leq \abs{\widehat{\mu} - \mu} \leq 2\eps'$,} \\
        \text{with $O(\log(1/\delta)/\eps')$ queries we can distinguish, except with probability at most~$\delta$.} \label{eqn:guar}
    \end{multline}
    Note this is of the form needed for the log\,log trick, with the ``accuracy'' parameter being~$\eps'$. We will be repeatedly using \Cref{eqn:guar} in a kind of binary search, with a sequence of $\eps'$ values starting at~$1$, decreasing by a factor of $1-c$ or less at each stage, and terminating with a value at least~$\eps/2$.  Thus the log\,log trick tells us the final query complexity will be~$O(1/\eps)$, as desired.  

    Our binary search will aim to ensure that in its $j$th stage, $\mu$~is guaranteed to be in the interval~$I_j = [a_j, b_j]$.
    We may start with $I_1 = [-1,1]$; this is guaranteed to contain~$\mu$ because 
    \begin{equation}
        \mu^2 = {\E}_p[\by]^2 \leq {\E}_p[\by^2] = s^2 \leq 1,
    \end{equation}
    the last inequality by the promise of \Cref{prob:problem3}.

    In the $j$th stage of the binary search, we employ~\eqref{eqn:guar}, with
    \begin{equation}
        \eps' = \eps_j \coloneq |I_j|/2 \quad\text{and}\quad\wh{\mu} = a_j + c\eps_j.
    \end{equation}
    As mentioned, the initial $\eps'$ value is~$1$.
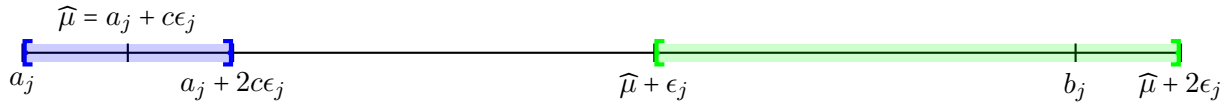
\begin{figure}[h]
    \centering

    \begin{tikzpicture}[scale=7]
\draw[-, thick] (-1,0) -- (1.2,0);
        \foreach \x/\xtext in {-1/$a_j$,-0.8/,-0.6/$a_j+2c\eps_j$,0.2/$\wh{\mu}+\eps_j$,1/$b_j$,1.2/$\wh{\mu}+2\eps_j$}
            \draw[thick] (\x,0.5pt) -- (\x,-0.5pt) node[below] {\xtext};
        \draw (-0.8,0.5pt) node[above] {$\wh{\mu}=a_j+c\eps_j$};
        \draw[[-, ultra thick, blue] (-1,0) -- (-.99,0);
        \draw[{-]}, ultra thick, blue] (-0.61,0) -- (-0.6,0);
        \draw[{[-}, ultra thick, green] (0.2,0) -- (0.21,0);
        \draw[{-]}, ultra thick, green] (1.19,0) -- (1.2,0);
        \fill[opacity = 0.2, blue] (-1,-.10ex) -- (-0.6, -.10ex) -- (-0.6, .10ex) -- (-1,.10ex) -- cycle;
        \fill[opacity = 0.2, green] (0.2,-.10ex) -- (1.2, -.10ex) -- (1.2, .10ex) -- (0.2,.10ex) -- cycle;
\end{tikzpicture}
    \caption{In each step of the binary search procedure, we eliminate either the left interval  or the right interval from consideration, which reduces the interval we're left with by a constant factor.}
\end{figure}

    Assuming (by virtue of the log\,log trick) that all results obtained from~\eqref{eqn:guar} are correct, let us describe how~$I_2, I_3, \dots$ may be chosen.
    \begin{itemize}
        \item Supposing that in the $j$th stage we get from~\eqref{eqn:guar} that $\eps_j \leq \abs{\widehat{\mu} - \mu} \leq 2\eps_j$.
        Then it must be that $\abs{\widehat{\mu} - \mu} \not \leq c \eps_j$.  Thus we may take
            \begin{equation}
                I_{j+1} = [a_{j} + 2c\eps_j, b_j] = [a_j + c|I_j|, b_j].
            \end{equation}
        \item Supposing that in the $j$th stage we get from~\eqref{eqn:guar} that $\abs{\widehat{\mu} - \mu} \leq c\eps_j$.
        Then it must be that $\eps_j \leq \abs{\widehat{\mu} - \mu} \leq 2\eps_j$ fails to hold.
        It can't be that $\abs{\widehat{\mu} - \mu} > 2\eps_j = |I_j|$, and hence it must be that $\abs{\widehat{\mu} - \mu} < \eps_j$.
        Thus we may take
        \begin{equation}
            I_{j+1} = [a_{j}, a_j + (3/4)|I_j|] \supseteq [a_j, a_j + (c + 1/2)|I_j|] = [a_j, \wh{\mu} + \eps_j].
        \end{equation}
    \end{itemize}
    In either case, observe that $|I_{j+1}| \leq (1-c)|I_j|$. That is, the widths of our intervals become smaller by a factor of $1-c$ or less at each stage, as promised for the log\,log trick.
    We may terminate the search once we reach some~$I_{t+1}$ with $|I_{t+1}| \leq \eps$; thus indeed the final use of \Cref{eqn:guar} has accuracy parameter at least~$\eps/2$ so we get total query cost~$O(1/\eps)$.
\end{proof}

Our next upgrade will be to achieve the result of~\cite{BHT98} from \Cref{table:prior}, namely optimal mean estimation for Bernoulli random variables.
This uses the standard trick of ``successive halving''.
Besides illustrating that it can be achieved via our methods, we will actually \emph{need} this result as a lemma for our final mean estimation algorithm.  
\begin{prob}[Bernoulli mean estimation] \label{prob:problem4a}
    \begin{tabbing}
      Input: \quad\= Positive integer~$n$.\\
      Promise: \> $\by \in \{0,1\}$ always.\\
      Output: \> An estimate $\widehat{\mu}$ such that $\abs{\widehat\mu - \mu} \leq \sigma/n$.
    \end{tabbing}
\end{prob}
\begin{lemma}   \label{lem:sim}
    We can solve \Cref{prob:problem4a} with $O(n)$ queries  to the code for~$\by$.
\end{lemma}
\begin{proof}
    Let $p = \Pr[\by = 1]$, so $\mu = p$ and $\sigma = \sqrt{p(1-p)}$.
    It is convenient to first reduce to the case of small~$p$, say $p \leq 3/4$.
    We can do this using~$O(1)$ queries by employing our solution to \Cref{prob:problem3} with~$\eps = 1/4$. 
    If this yields an estimate~$\wh{\mu} \geq 1/2$, then we can be confident~$\mu  = p \geq 1/4$. 
    In this case, we simply replace~$\by$ with~$1-\by$ and subtract our final estimate from~$1$; this leaves~$\sigma$ unchanged and achieves~$p \leq 3/4$.
    Note that $\sigma \geq \frac12 \sqrt{p}$, so it suffices for our algorithm to estimate~$p$ to within an additive~$\sqrt{p}/(2n)$.
    Because of this, if ever the algorithm determines that $p \leq 1/(4n^2)$, it may acceptably output the estimate~$\wh{\mu} = 0$. 

    Our algorithm now proceeds in stages, always maintaining an upper bound~$\ol{p}$ on~$p$.
    Initially, $\ol{p} = 3/4$.  
    So long as $\ol{p} > 1/(4n^2)$, the algorithm attempts to lower~$\ol{p}$ by a constant factor.
    It does this by applying the algorithm for \Cref{prob:problem3} on the random variable $\by' \coloneqq \by/\sqrt{\ol{p}}$, with its error parameter~$\eps$ set to~$\tfrac14 \sqrt{\ol{p}}$.
    We will later observe that the log\,log trick applies, and for now assume all estimates are accurate. 
    So given an estimate~$\wh{\mu}'$ of~$\E_p[\by']$ that is correct to an additive~$\tfrac14 \sqrt{\ol{p}}$, multiplying it by~$\sqrt{\ol{p}}$ gives an estimate~$p'$ of~$p$ that is correct to an additive $\tfrac14{\ol{p}}$.
    
    If $p' \leq \tfrac12{\ol{p}}$, the algorithm may infer that~$p \leq \tfrac34{\ol{p}}$, and therefore lower~$\ol{p}$ by a factor of~$\frac34$ for the next stage.  
    On the other hand, if $p' \geq \tfrac12{\ol{p}}$, the algorithm may infer that $p \geq \tfrac14{\ol{p}}$, and thus $\tfrac12{\ol{p}}$ is within a factor~$2$ of~$p$. 
    
    The algorithm proceeds in this way until either $\ol{p} \leq 1/(4n^2)$ (at which point it may safely output~$\wh{\mu} = 0$) or else it knows a factor-$2$ approximation~$\wh{p} \geq 1/(8n^2)$ of~$p$.
    In the latter case, the algorithm uses the solution to \Cref{prob:problem3} one more time, on the random variable $\by/\sqrt{2\wh{p}}$ (which has second moment at most~$1$, as needed), with error parameter $\eps = \frac{1}{2n}$.
    This requires $O(n)$ uses of the code for~$\by$, and --- multiplying the estimate by~$\sqrt{2\wh{p}}$ --- yields an estimate for~$p$ that is within additive error $\sqrt{2\wh{p}} \cdot \frac{1}{2n}  \leq \sqrt{p}/n = \sigma/n$, as desired.

    It remains to remark that we can use the log\,log trick as before to ensure high confidence in all stages succeeding; when our current bound on~$p$ is $\ol{p}$, we can define the ``accuracy parameter'' to be $\eps = \tfrac14 \sqrt{\ol{p}}$.
    Then as in the preceding proof, we can achieve this accuracy and~$\delta$ confidence using $O(1/\eps) \log(1/\delta)$ queries.
    Since $\ol{p}$ decreases by a factor of~$3/4$ in each stage, the accuracy parameter decreases by a factor~$\sqrt{3/4} < 1$.
    And since $\ol{p}$ never goes below~$1/(4n^2)$, our final accuracy parameter~$\eps^*$ may be set to~$1/(8n)$, meaning the total query cost will be~$O(n)$, as desired.
\end{proof}

Next we observe that the result of~\cite{Ter99} from \Cref{table:prior} follows almost immediately:
\begin{prob}[Mean estimation (suboptimal) for bounded random variables] \label{prob:problem4b}
    \begin{tabbing}
      Input: \quad\= Positive integer~$n$.\\
      Promise: \> $\by \in [0,1]$ always.\\
      Output: \> An estimate $\widehat{\mu}$ such that $\abs{\widehat\mu - \mu} \leq \sqrt{\mu}/n$.
    \end{tabbing}
\end{prob}
\begin{lemma}   \label{lem:ter}
    We can solve \Cref{prob:problem4b} with $O(n)$ queries  to the code for~$\by$.
\end{lemma}
\begin{proof}
    We reduce from the $\{0,1\}$-valued case essentially as Terhal~\cite{Ter99}.
    Given the code for a random variable~$\by \in [0,1]$, we can tack on 
    a small amount of additional classical randomness, forming code for a related random variable~$\by' \in \{0,1\}$ as follows: Draw $\by$, and if the outcome is~$y$, let~$\by'$ be a $\{0,1\}$-valued random variable with expectation~$y$.
    In this way, $\E_p[\by'] = \mu$, and $\stddev_p[\by'] \leq \sqrt{\E_p[(\by')^2]} = \sqrt{\mu}$.
    Thus we may apply our solution to \Cref{prob:problem4a} to~$\by'$ to complete the proof.
\end{proof}
 
Finally we show how to achieve the results due to \cite{Hei02,Mon15} from \Cref{table:prior} with query complexity~$O(n)$: essentially, optimal Mean Estimation in terms of a \emph{known} upper bound on the standard deviation.
Aside from a trivial scaling issue, the difference between this and our \Cref{prob:problem3} is that we only wish to assume a bound on~$\sigma$ rather than~$s$.
Since $\sigma^2 = \second^2 - \mu^2$, we can only have $\second$ significantly larger than $\sigma$ if $\mu$ is very large compared to~$\sigma$.  
As noted by Montanaro~\cite{Mon15}, such a situation can easily be fixed by subtracting one ``typical'' value of~$\by$ from each subsequent draw.
\begin{prob}[Mean estimation in terms of a standard deviation bound] \label{prob:problem5}
    \begin{tabbing}
      Input: \quad\= Positive integer~$n$, and parameter~$\sigma_{\text{bound}} \geq 0$.\\
      Promise: \> $\sigma \leq   \sigma_{\text{bound}}$.\\
      Output: \> An estimate $\widehat{\mu}$ such that $\abs{\widehat\mu - \mu} \leq \sigma_{\text{bound}}/n$.
    \end{tabbing}
\end{prob}
\begin{lemma}   \label{lem:mont}
    We can solve \Cref{prob:problem5} with $O(n)$ queries  to the code for~$\by$.
\end{lemma}
\begin{proof}
    The idea essentially appears in~\cite{Mon15}.
    If~$\sigma_{\text{bound}} = 0$ then~$\by$ is constant and one draw suffices to get~$\mu$ exactly. 
    Otherwise, by scaling we may assume that~$\sigma_{\text{bound}}$ is, say,~$1/4$; then our  target additive error is~$1/(4n)$.
    
    The algorithm first uses the code for~$\by$ to draw a single sample --- call the sample $\boldm$, and say its outcome is~$m$.
    Then the algorithm forms (the code for) a new random variable~$\by' = \by - m$.
    We have $\mu' \coloneqq \E_p[\by'] = \mu - m$, so it suffices to estimate~$\mu'$ to an additive~$1/(4n)$.

    By applying Chebyshev's inequality to~$\boldm$, we get that that $\abs{m - \mu} \leq 2\sigma$ except with probability at most~$1/4$.  
    Assuming this happens, we have
    \begin{multline}
        (s')^2 \coloneqq {\E}_p[(\by')^2] = {\E}_p[(\by - \mu + \mu - m)^2] \leq 2{\E}_p[(\by - \mu)^2] + 2{\E}_p[(\mu - m)^2] \\
        \leq 2\sigma^2 + 2(2\sigma)^2 = 10 \sigma^2 \leq 10 \sigma_{\text{bound}}^2 = 10/4^2 \leq 1.
    \end{multline}  
    Thus the promise of \Cref{prob:problem3} is satisfied for~$\by'$, and by taking~$\eps = 1/(4n)$ (and repeating our algorithm for \Cref{prob:problem3} a few times to get failure probability at most~$1/12$), we get the necessary estimate for~$\mu'$.
\end{proof}

\subsection{The final upgrade: handling an unknown standard deviation}
With our solution to \Cref{prob:problem5} in hand, the last remaining challenge is to avoid assuming a known upper bound on the standard deviation~$\sigma$ of~$\by$.

To begin, we return to our solution to \Cref{prob:problem4b} concerning $[0,1]$-valued random variables.
To make it look more like our final goal, we achieve error~$s/n = \sqrt{\E_p[\by^2]}/n$ rather than the larger $\sqrt{\mu} = \sqrt{\E_p[\by]}/n$.
However we will have to assume that~$s \geq 1/n$.
\begin{prob} \label{prob:problem6}
    \begin{tabbing}
      Input: \quad\= Positive integer $n$.\\
      Promise: \> $\by \in [-1,1]$ and $\second \geq 1/n$.\\
      Output: \> An estimate $\widehat{\mu}$ such that $\abs{\widehat\mu - \mu} \leq \second/n$.
    \end{tabbing}
\end{prob}
\begin{lemma}
    We can solve \Cref{prob:problem6} with $O(n)$ queries  to the code for~$\by$.
\end{lemma}
\begin{proof}
    The first step is to \emph{estimate}~$\second$ to within a factor of~$2$.
    To do this, we apply our solution to \Cref{prob:problem4b} to the random variable~$\bz \coloneqq \by^2$.  
    We have $\bz \in [0,1]$ since $\by \in [-1,1]$, so with $O(n)$ queries we can get an additive estimate of $\E_p[\bz] = \E_p[\by^2] = s^2$ that is correct to within an additive~$\sqrt{s^2}/(2n) = s/(2n)$.
    But $s/(2n) \leq s^2/2$ since we have the promise $s \geq 1/n$.
    Thus our  estimate of~$s^2$ is within an additive~$s^2/2$; i.e., it is a factor-$2$ multiplicative estimate. 
    So we have a factor-$2$ (or even~$\sqrt{2}$) multiplicative estimate~$\hat{s}$ of~$s$.
    
    Given this, we can form the rescaled random variable~$\by' \coloneqq \by/(2\hat{s})$, which has $(s')^2 \coloneqq \E_p[(\by')^2] \leq 1$.
    Then applying our solution to \Cref{prob:problem3}
    with error parameter $\eps = 1/(4n)$, we use $O(n)$ queries to get an estimate $\wh{\mu}'$ of $\E_p[\by] = \mu/(2\hat{s})$ that is correct to an additive $1/4n$.
    Finally, taking $\wh{\mu} = (2\hat{s}) \cdot \wh{\mu}'$, we have an estimate of~$\mu$ that is correct to an additive $(2\hat{s})/(4n) \leq s/n$, as desired. 
\end{proof}

We now come to the (almost-final) step: using the quantum Quantile Estimation algorithm of Hamoudi~\cite{Ham21}.
With $O(n)$ queries, this will allow us to find a suitable ``cap'' value~$B$ such that replacing~$\by$ with its truncation~$\by'$ to the interval~$[-B,B]$ does not substantially change the mean estimation task.
As long as we have $B \leq n \cdot \sqrt{\E_p[(\by')^2]}$, we will be able to employ our solution to \Cref{prob:problem6} (after dividing~$\by'$ by~$B$).

The correct value to choose for~$B$ is, roughly speaking, the ``$(1-1/n^2)$''-quantile value for~$\by$; i.e., the largest~$B$ such that $\Pr[\by \geq B] \geq 1/n^2$.
Hamoudi's algorithm can find this~$B$ with $O(n)$ samples from~$\by$. 
(Classically, we could find this~$B$ by taking~$\Theta(n^2)$ draws from~$\by$ and and outputting the maximum sample seen. 
The intuition for Hamoudi's algorithm is to take this and apply the square-root quantum speedup afforded by the Minimum Finding algorithm of~\cite{DH96}.)

On one hand, a Chebyshev-type argument shows that if~$B$ is so large that $\Pr[\by \geq B] \ll 1/n^2$, then capping~$\by$ at~$B$ does not affect the mean/second-moment enough to make a substantial difference to the mean estimation problem.   On the other hand, this value of~$B$ will be small enough that the $B \leq n \cdot \sqrt{\E_p[(\by')^2]}$ required for \Cref{prob:problem6} holds. 
This is because (roughly speaking) we have $\Pr[\by \approx B] \gtrapprox 1/n^2$  (else the quantile value~$B$ could be chosen larger), and hence even the capped~$\by'$ will have $\Pr[\by' \approx B] \gtrapprox 1/n^2$, implying $\E_p[(\by')^2] \gtrapprox B^2/n^2$.

\begin{prob}[Mean estimation in terms of second moment] \label{prob:problem7}
    \begin{tabbing}
      Input: \quad\= Positive integer $n$.\\
      Output: \> An estimate $\widehat{\mu}$ such that $\abs{\widehat\mu - \mu} \leq \second/n$.
    \end{tabbing}
\end{prob}
\begin{lemma}
    We can solve \Cref{prob:problem7} with $O(n)$ queries to the code for~$\by$.
\end{lemma}
\begin{proof}
    We begin by performing the Quantile Estimation algorithm of Hamoudi~\cite{Ham21} on the random variable~$\abs{\by}$. This uses $O(n)$ queries and (with high probability) determines a number~$B$ (a possible outcome for~$\abs{\by}$) such that:
    \begin{equation}    \label[ineq]{ineq:ham}
        {\Pr}_p[\abs{\by} \geq B] \geq 1/n^2, \qquad {\Pr}_p[\abs{\by} > B] \leq C/n^2.
    \end{equation}
    (Here $C$ is a large universal constant, and Hamoudi's analysis also requires that~$n$ is at least some universal~$n_0$ --- but  we may freely assume that.)
    
    We now follow Hamoudi's idea and define 
    \begin{equation}
        \by' = \begin{cases}
            \phantom{+}\by & \text{if $\abs{\by} \leq B$,} \\
            +B & \text{if $\by > B$,}\\
            -B & \text{if $\by < -B$.}
        \end{cases}
    \end{equation}
    Our first goal is to get a good estimate for $\mu' \coloneqq \E_p[\by']$.
    If $B = 0$ then we immediately know~$\mu' = 0$, since~$\by' \equiv 0$.
    Otherwise, let us consider the random variable $\by'' \coloneq \by'/B \in [-1,1]$.  
    We know 
    \begin{equation}
        (s')^2 \coloneqq {\E}_p[(\by')^2] \geq B^2 \cdot {\Pr}_p[\abs{\by'} \geq B] = B^2 \cdot {\Pr}_p[\abs{\by} \geq B] \geq B^2/n^2,
    \end{equation}
    where we used the first \Cref{ineq:ham}. 
    Thus  $s'' = s'/B \geq 1/n$, meaning the promises of \Cref{prob:problem6} are  satisfied for~$\by''$.
    Thus with $O(n)$ we can obtain an estimate for~$\mu''$ to within~$s''/n$, hence an estimate~$\wh{\mu}$ for~$\mu'$ within~$s'/n$.
    Since $\by'$ is a truncation of~$\by$, we clearly have $s' \leq s$; thus $\abs{\wh{\mu} - \mu'} \leq s/n$.

    It now suffices to show the claim $|\mu - \mu'| \leq O(s/n)$; this will imply $|\wh{\mu} - \mu| \leq O(s/n)$, and we can complete the proof of the lemma by adjusting~$n$ by a constant factor.
    
    To show this the claim, observe that 
    \begin{equation}
        \abs{\mu - \mu'} = \abs{{\E}_p[\by -\by']}
        \leq {\E}_p[\abs{\by - \by'}] \leq {\E}_p\bracks*{\Bigl|\abs{\by} - B\Bigr| \cdot 1_{\{\abs{\by} > B\}}},
    \end{equation}
    as $\by' = \by$ when $\abs{\by} \leq B$. 
    Now Cauchy--Schwarz implies the above is at most
    \begin{equation}
        \sqrt{{\E}_p[(\abs{\by}-B)^2 \cdot 1_{\{\abs{\by} > B\}}]} \sqrt{{\E}_p[1_{\{\abs{\by} > B\}}]} \leq \sqrt{\E[\abs{\by}^2]} \cdot \sqrt{{\Pr}_p[\abs{\by} > B]} \leq \sqrt{C} \cdot \second/n,
    \end{equation}
    where we used the second \Cref{ineq:ham}.
    Thus we have established the claim $\abs{\mu - \mu'} \leq O(s/n)$, completing the proof.
\end{proof}

Finally, we come to the main Mean Estimation problem; its only difference from \Cref{prob:problem7} is that it has the standard deviation~$\sigma$ in place of~$s$: 
\begin{prob}[Mean Estimation]   \label{prob:meanestimationmain}
    \begin{tabbing}
        Input: \quad\= Positive integer $n$.\\
        Output: \> An estimate $\widehat{\mu}$ such that $\abs{\widehat\mu - \mu} \leq \sigma/n$.
      \end{tabbing}
\end{prob}
Our main \Cref{thm:main0} is equivalent to saying that \Cref{prob:meanestimationmain} can be done with~$O(n)$ queries.
But this follows from our solution to \Cref{prob:problem7} via the Montanaro trick, in exactly the same way that \Cref{lem:mont} follows from \Cref{lem:us}.

\bibliographystyle{alphaurl}
\bibliography{quantum}

\appendix

\section{Gate complexity}   \label{sec:gates}

In this section we sketch how to establish \Cref{rem:gatecxty}, that our algorithm's gate complexity is (essentially) minimal given its sample complexity: namely, $O(nS)$, where $S$ is the gate complexity of ``the code'' for~$\by$.
The only potential excess comes from having to classically compute the $\arctan$ function. 
More precisely, we show the gate complexity is
\begin{equation}
    O(nS) + O(n \log n \cdot (\log \log n)^2).
\end{equation}
\begin{remark}
    The extra additive term $O(n \log n \cdot (\log \log n)^2)$ above can be absorbed into the~$O(nS)$ except when $S < o(\log n \cdot (\log \log n)^2)$.
    On the other hand, if $S < o(\log n)$, with gate complexity~$n^{o(1)}$ we can compute~$\E_p[\by]$ \emph{exactly}\footnote{Exactly, if the gates used to compute~$\by$ have amplitudes that are exactly representable.  Otherwise, up to $O(\log n)$ bits of precision --- which suffices, as we will explain.} by brute-force analysis of all computational paths in the circuit for~$\by$.
    Thus only in the unusual case of $\Omega(\log n) \leq S < o(\log n \cdot (\log \log n)^2)$ must we report our algorithm's gate complexity as~$o(nS \cdot (\log \log n)^2)$, rather than~$O(nS)$.
\end{remark}

\paragraph{Assumptions.} We use the standard model of CNOT gates together with any $1$-qubit gate.
(From these one can also build Toffoli gates~\cite[Fig.~4.9]{NC10}.)
We will assume that the code for~$\by$ outputs its value in a ``floating point'' format (of at most~$S$ bits). Hence given output values of the code, we can perform the following with gate complexity~$O(S)$: subtraction, comparison with~$0$, rounding to a power of~$2$, and multiplication/division by a power of~$2$ (shifting).

\subsection{Summary of the steps of the algorithm}
Here we summarize the algorithms needed for Mean Estimation with error~$\sigma/n$.

\paragraph{Solving \Cref{prob:problem3}.} (I.e., mean estimation for random variables~$\by$ satisfying~$\E_p[\by^2] \leq 1$.)
This algorithm will always be run with precision parameter $\eps \geq \Omega(1/n)$:
\begin{itemize}
    \item Binary search for~$\mu$ with intervals of of width decreasing geometrically from~$2$ to~$\Omega(\eps)$.
    \item Test each interval centered at~$\wh{\mu}$ by replacing~$\by$ with $\by - \wh{\mu}$ and performing our solution to the Main Task, namely:
    \begin{itemize}
        \item Converting the code for~$\by$ to (controlled versions of) $\Refl$ and~$\RotF$.
        \item Performing Quantum Phase Estimation on~$\calU = \Refl \cdot \RotF$.
    \end{itemize}
\end{itemize}

\paragraph{The overall Mean Estimation algorithm.}
This is obtained by reading \Cref{sec:consequences} roughly backward:
\begin{itemize}
    \item Draw one sample~$m$ from~$\by$ and replace~$\by$ with $\by - m$. (This is to go from \Cref{prob:meanestimationmain} to \Cref{prob:problem7}.)
    \item Perform Hamoudi's Quantile Estimation on~$\abs{\by}$ obtaining~$B$; replace~$\by$ with its truncation to~$[-B,B]$ and divide it by~$B$.
    Call the resulting random variable~$\ol{\by}$.
    \item Estimate~$\E[\bz]$ for $\bz = \ol{\by}^2$ to factor~$2$.  This uses the solution to \Cref{prob:problem4b}, as follows:
    \begin{itemize}
        \item Replace the $[0,1]$-valued~$\bz$ with a randomized $\{0,1\}$-valued version.
        \item Starting with a trivial upper bound $\ol{p}$ for $p = \Pr[\bz = 1]$, repeatedly use the solution to \Cref{prob:problem3} to estimate $\ol{p}$ to within $\eps \approx \frac{1}{4}\sqrt{\ol{p}}$, and lower~$\ol{p}$ if necessary (but not below $1/(4n^2)$)
    \end{itemize}
    \item Having determined $\E_p[\ol{\by}^2]$ to factor~$2$, rescale it so that~$\E[\ol{\by}^2] \approx 1$; then use the solution to \Cref{prob:problem3} on~$\ol{\by}$.
\end{itemize}

\subsection{Precision issues}
Suppose we have done the first two steps of the overall Mean Estimation algorithm, obtaining $m$~and~$B$; these numbers are expressed in floating point with at most~$S$ bits.
All subsequent stages of the algorithm work with the $[-1,1]$-bounded random variable $\ol{\by} =\frac1B \cdot \mathrm{trunc}_{[-B,B]}(\by - m)$.
In gate complexity $O(S)$ it is easy to compute~$\mathrm{trunc}_{[-B,B]}(\by - m)$.
Dividing by~$B$ is not as easy, but we argue that it is fine if the algorithm simply divides by the next largest power of~$2$, call it~$B'$. 
The only properties we needed in from this scaling in \Cref{prob:problem6} were that the resulting random variable has~$\ol{\by} \in [-1,1]$ and $\sqrt{\E[\ol{\by}^2]} \geq 1/n$.
If we divide by a slightly larger~$B'$, the first property still holds; and while the second is no longer literally true, it is true up to a factor of~$4$, we we can easily compensate for by adjusting the constant factor on~$n$.
Thus once we are dividing by~$B'$ --- a power of~$2$ --- (and also subsequently multiplying our final estimate by~$B'$), the gate complexity becomes~$O(S)$ as this just amounts to bit-shifting. 

Thus we have argued so far that we can still get one sample from~$\ol{\by} \in [-1,1]$ with gate complexity~$O(S)$.
We now wish to argue that samples and computations done with~$\ol{\by}$ can be rounded to $O(\log n)$ bits of precision.  
(Recall also that $S \geq \Omega(\log n)$ without loss of generality.)
To see this, it suffices to note that all the remaining steps of the algorithm only care about numbers and interval widths that are at least~$\Omega(1/n^2)$; hence round-off to a sufficiently large $O(\log n)$ bits will affect their accuracy/correctness by at most small constant factors.
As usual, we can ultimately compensate for these small constant factors by increasing our sample complexity by a constant factor.

\paragraph{Precision conclusion:} For the post-Hamoudi part of our algorithm, we can assume the gate complexity of obtaining a sample is still~$O(S)$, and then that all subsequent numbers and  computations require bit complexity only~$O(\log n)$.

\subsection{Final gate analysis}

We now analyze the steps of the algorithm to justify our claim that the overall gate complexity is $O(nS)$ plus $O(n) \cdot \wt{O}(\log n)$.

The first step of sampling and subtracting~$m$ does not cost more than~$O(S)$ gates per sample, so the first serious piece of the algorithm to analyze is Hamoudi's Quantile Estimation algorithm.

\paragraph{Hamoudi's algorithm.}
Hamoudi does not explicitly analyze the gate complexity of his algorithm; he just shows the sample complexity is~$O(n)$.
We argue that the gate complexity is $O(nS) + O(n) \cdot \wt{O}(\log n)$.
To do this, we first sketch his algorithm as applied to a random variable~$\bx$.
In short, the algorithm produces a sequence $-\infty = x_0 < x_1 < x_2 < \cdots$ by repeatedly setting~$y_{t+1}$ to be a draw from $\bx \mid (\bx > x_t)$.
Obtaining each draw is done via a ``Sequential Amplitude Estimation'' quantum algorithm, as in~\cite{BBHT98} (similar to our \Cref{lem:sim}).
In turn, the pseudocode for this is roughly the following:
\begin{enumerate}
    \item For $j = 1, 2, 3, \dots$
    \item \qquad Let $T_j$ be a random integer in $[(6/5)^{j-1}, (6/5)^{j}]$.
    \item  \qquad Do Amplitude Amplification~\cite{BHMT00} with~$O(T_j)$ applications of the ``code for~$\bx$''. More precisely, the synthesizer for~$\bx$ and its inverse are used~$O(T_j)$ times, as is~$1-2\Pi$, where~$\Pi$ is a projector depending on the code~$\calX$ for~$\bx$ and comparison with the current~$x_t$.  (This comparison takes only~$O(S)$ gates.)
    \item  \qquad Measure with $(\Pi, \Id - \Pi)$ and quit the ``for loop'' if~$\Pi$ occurs.
    \item Do a measurement to obtain the next~$x_t$.
\end{enumerate}
Moreover, throughout this pseudocode the algorithm maintains a counter of the number of times the code for~$\bx$ has been applied, and it halts (with the final value of~$x_t$ as its output) once the budget of $O(n)$ samples has been hit.

As mentioned, we claim that this algorithm can be implemented with a circuit of size~$O(nS) + O(n) \cdot \wt{O}(\log n)$.  Essentially, we can unroll all loops and make one sequence of~$O(n)$ applications of (pieces of) the circuit for~$\bx$. Interspersed between of applications of the circuit for~$\bx$ will be conditionals and operations operating on~$T_j$ (which is $O(\log n)$ bits), comparisons of~$O(S)$-bit quantities, and conditional measurements (including measuring qubits to get classical random bits).  These measurements can be moved to the end by the usual Principle of Deferred Measurements.

The last output~$x_t$ is the value of ``$B$''.
As discussed before, with~$B$ in hand, we can assume that all future calculations done on draws from~$\ol{\by}$ are done with~$O(\log n)$ bits of precision.

\paragraph{Estimating $\E[\bz] = \E[\ol{\by}^2]$ to factor~$2$.}
The next step of the algorithm, estimating the mean of~$\ol{\by}^2$ to a factor of~$2$, first involves replacing~$\ol{\by}$ by~$\bz = \ol{\by}^2$.  
In turn this means we have to multiply an~$O(\log n)$-bit sample by itself, creating an overhead of $M(\log n) = O(\log n \log\log n)$ gates per sample~\cite{HH21}.

Next, in \Cref{lem:ter}, we must implement the additional randomness for converting a $[0,1]$-valued random variable to a $\{0,1\}$-valued one.
This involves taking a sample outcome~$z$ (expressed with~$O(\log n)$ bits) and producing $\sqrt{1-z}\ket{0} + \sqrt{z} \ket{1}$.
In turn, this involves $O(\log n)$ controlled rotations, plus the computation of~$\arctan \sqrt{z/(1-z)}$ (to $O(\log n)$ bits of precision).
The gate complexity of this computation is $O(M(\log n) \cdot \log \log n) = O(\log n \cdot (\log \log n)^2)$~\cite{BZ11,HH21}.
Thus now one ``sample'' costs gate complexity $O(S + \log n \cdot (\log \log n)^2)$; this (along with a subsequent $\arctan$ computation) is the computational ``bottelneck'' leading to our final gate complexity of~$n$ times the cost of a sample.

To complete the estimate of~$\E[\bz]$, we need to do the perform the ``successive halving'' routine of~\Cref{lem:sim}.
This mainly uses our solution to \Cref{prob:problem3}, which we analyze below.  
The only other aspect is the log\,log trick --- for which the appropriate~$\delta_j$ values can be precalculated --- and minor computations on~$O(\log n)$ bit numbers (which are within our budget).

It remains to analyze the gate complexity of our solution to \Cref{prob:problem3}.

\paragraph{The binary search.}
In the binary search, it is easy to adjust constants so that the intervals~$I_j$ decrease in width by precisely some factor $1-c'$ at each stage.  Thus the whole binary search can be straightforwardly unrolled with the interval widths and the $\delta_j$'s precalculated.
The only additional work that needs to be done is counting and comparing $O(\log n)$-bit integers (to compute majorities).
Again, this portion of the algorithm only incurs an additive overhead of $O(n) \cdot \wt{O}(\log n)$.

Finally, we reach:

\paragraph{The Main Task.}  Finally we come to our algorithm for the Main Task from \Cref{sec:mainthing}.
At this point, each of our samples costs $O(S) + \wt{O}(\log n)$ gates, we need to do all calculations with precision $O(\log n)$ bits, and  the sample complexity is~$O(1/\eps) = O(n)$.
The gate complexity of phase estimation is the cost to prepare the initial state ($O(S)$ for us), plus the sample complexity times the cost of controlled-$\calU$, plus the cost of some minor calculations on numbers of $O(\log (1/\eps)) = O(n)$ bits.  
Since our final claimed gate complexity is $O(nS) + O(n \log n \cdot (\log \log n)^2)$, it now remains to argue that the cost to compute controlled-$\calU$ is $O(S) + O(M(\log n) \log \log n)$.

The cost to compute controlled-$\calU$ is the cost to compute controlled-$\Refl$ and controlled-$\RotF$.
The former has gate complexity~$O(S)$.
As for the latter, we need to take a sample~$y$, compute $\alpha = -2\arctan y$ (to $O(\log n)$ bits of precision), and apply $O(\log n)$ controlled-phase gates to implement the phase $e^{\i \alpha}$. 
Similarly to before, the main bottleneck is computing the $\arctan$, and the gate complexity is
$O(M(\log n) \cdot \log \log n) = \wt{O}(\log n)$~\cite{BZ11,HH21}.

\end{document}